\RequirePackage{fix-cm}
\documentclass{svjour3}                     
\smartqed  

\usepackage{fullpage,amsmath,amssymb, amsfonts, url,graphicx,latexsym}
\usepackage[nothing]{algorithm}
\usepackage{algorithmicx}
\usepackage{setspace}
\usepackage[noend]{algpseudocode}

\floatname{algorithm}{Pseudocode}

\usepackage{xcolor}

\newcommand{\lastcorrections}%
{{
 \begin{sloppypar}
    \baselineskip -0.2in
    \tiny\bf\noindent
last corrections:\\
\end{sloppypar}
}}

\newcommand{\myparagraph}[1]{{\smallskip\noindent{\bf #1}}}
\newcommand{\emparagraph}[1]{{\smallskip\noindent\emph{#1}}}

\newcommand{\barlambda}{{\bar \lambda}}

\newcommand{\calJ}{{\cal J}}

\newcommand{\starf}{f^\ast}
\newcommand{\starh}{h^\ast}
\newcommand{\starH}{H^\ast} 
\newcommand{\starlambda}{\lambda^\ast}

\newcommand{\malymax}{{\mbox{\tiny\rm max}}}

\newcommand{\circh}{h^\circ}
\newcommand{\circH}{H^\circ}

\newcommand{\braced}[1]{{ \left\{ #1 \right\} }} 
\newcommand{\smallbraced}[1]{{\{ #1 \}}}

\newcommand{\brackd}[1]{{ \left[ #1 \right] }}

\newcommand{\floor}[1]{{ \lfloor #1 \rfloor }}
\newcommand{\ceiling}[1]{{ \lceil #1 \rceil }}

\newcommand{\suchthat}{{\;:\;}}
\newcommand{\assign}{\,{\leftarrow}\,}

\newcommand{\MGaps}{\textsf{D}}

\newcommand{\TotalFlow}{\textsf{F}}
\newcommand{\Thput}{\textsf{T}}
\newcommand{\rmin}{r_{\textrm{min}}}
 
\newcommand{\total}{{\Sigma}}

\newcommand{\MaxGaps}{{\mbox{\sc MaxGaps}}}
\newcommand{\MaxThrpt}{{\mbox{\sc MaxThrpt}}}
\newcommand{\AlgViable}{{\mbox{\sc Viable}}}
\newcommand{\AlgMinMaxGap}{{\mbox{\sc MinMaxGap}}}
\newcommand{\AlgMinTotFlow}{{\mbox{\sc MinTotFlow}}}
\newcommand{\AlgMinGapMaxFlow}{{\mbox{\sc MinGapMaxFlow}}}
\newcommand{\AlgAdjust}{{\mbox{\sc Adjust}}}  
\newcommand{\AlgMinMaxflowGap}{{\mbox{\sc MinMaxflowGap}}}  
\newcommand{\AlgMatrixSelect}{{\mbox{\sc MatrixSelect}}}

\newcommand{\TRUE}{{\mbox{\textsf{true}}}}
\newcommand{\FALSE}{{\mbox{\textsf{false}}}}

\newcommand{\myIf}{\mbox{\textbf{if}}}
\newcommand{\myThen}{\mbox{\textbf{then}}}
\newcommand{\myReturn}{\mbox{\textbf{return}}}

\newenvironment{bigeqn*}{\large\begin{eqnarray*}}{\end{eqnarray*}}

\begin{document}

\title{Scheduling with Gaps: New Models and Algorithms\thanks{M.Chrobak was
supported by NSF grants CCF-0729071, CCF-1217314 and CCF-1536026.
M.~Golin was supported by grant FSGRF14EG28. }} 

\author{Marek Chrobak \and  Mordecai Golin \and  Tak-Wah Lam \and   Dorian Nogneng}

\institute{M.~Chrobak \at   
			Department of Computer Science,
         	University of California at Riverside, USA.
			\and
			M.~Golin \at 
			Department of Computer Science and Engineering,
			Hong Kong University of Science and Technology,
			Hong Kong.     
			\and
			T-W.~Lam \at   
			Department of Computer Science,
				University of Hong Kong,
				Hong Kong.
			\and
			D.~Nogneng \at
			LIX, {\'E}cole Polytechnique,
	 		France.
			}

\date{}  

\maketitle

\begin{abstract}
We consider scheduling problems for unit jobs with release times, where the number or size of the gaps
in the schedule is taken into consideration, either in the objective function or as a constraint.
Except for several papers on minimum-energy scheduling, there is no work in the scheduling literature that
uses performance metrics depending on the gap structure of a schedule.
One of our objectives is to initiate the study of such scheduling problems.

We focus on the model with unit-length jobs. First we examine scheduling problems with deadlines, where
we consider two variants of minimum-gap scheduling: maximizing throughput with a budget for
the number of gaps and minimizing the number of gaps with a throughput requirement. 
We then turn to other objective functions. For example, in some scenarios
gaps in a schedule may be actually desirable, leading to the  problem of maximizing the number of gaps. 
Other versions we study include minimizing maximum gap or maximizing minimum gap. 
The second part of the paper examines the model without deadlines, where we focus on 
the tradeoff between the number of gaps and the total or maximum flow time.

For all these problems we provide polynomial time algorithms, with running times ranging from
$O(n\log n)$ for some problems to $O(n^7)$ for other. The solutions involve a spectrum
of algorithmic techniques, including different dynamic programming formulations, speed-up techniques based
on searching Monge arrays, searching $X+Y$ matrices, or implicit binary search.

Throughout the paper, we also draw a connection between gap scheduling problems and their continuous analogues, namely
hitting set problems for intervals of real numbers. As it turns out, for some problems, the continuous
variants provide insights leading to more efficient algorithms for the corresponding
discrete versions, while for other problems completely new techniques are needed to solve the discrete version. 
\end{abstract} 


\section{Introduction}
\label{sec: introduction}
%
%



We consider scheduling of unit-length jobs with release times,
where the number or size of the gaps in the schedule is taken into
consideration, either in the objective function or as a constraint.  

This research was inspired by the work on scheduling problems where the objective is to minimize
the number of gaps in a schedule. Such problems arise in minimum-energy scheduling in the power-down
model, where a schedule specifies not only execution times of jobs but also
at what times the processor can be turned off. The processor uses energy at
rate $L$ per time unit when the power is on, and it does not consume any energy when 
it is off. If the energy required to power-up the system is less than $L$ 
then energy minimization is equivalent to minimizing the number of gaps in the schedule. 
The problem was introduced in 2005 by
Irani and Pruhs~\cite{Irani_Pruhs_algorithmic_05}, and its complexity remained open for a few years.
The first progress was achieved by Baptiste~\cite{Baptiste:min-idle-periods}, who gave 
a polynomial time algorithm for unit jobs that achieves running time $O(n^7)$. 
This time complexity was subsequently reduced to $O(n^4)$ in~\cite{Baptiste-etal-07,Baptiste_etal_polynomial_12}.
(In that paper a generalization to arbitrary processing times with job preemption is also considered.)
A greedy algorithm was analyzed in~\cite{CFHKLN_greedy_gaps_13,CFHKLN_greedy_gaps_17} and shown to have
approximation ratio $2$. Other variants of this problem have
been studied, for example the multiprocessor case~\cite{DeGHRZ07} or
the case when jobs have agreeable deadlines~\cite{Angel_etal_agreeable_energy_12,Angel_etal_agreeable_energy_14}.  
(See the survey in~\cite{Bampis_algorithmic_survey_16} for more information.)

To our knowledge, the above gap-minimization model is the only scheduling model
in the literature that considers gaps in the schedule as a  performance measure.
As we show, however, one can formulate a number of other natural, but not yet studied
variants of gap scheduling problems. Some of these problems can be
solved using dynamic-programming techniques resembling those used
for minimizing the number of gaps. Other require new approaches,
giving rise to new and interesting algorithmic problems.

Throughout the paper, we focus exclusively on the model with unit-length jobs.
The first type of scheduling problems we study involve jobs with release times and
deadlines. In this category, we address the following problems:
\begin{itemize}
	\item 
	In Section~\ref{sec: maximizing throughput with budget for gaps}, we study
	maximizing throughput (the number or total weight of scheduled jobs) with a budget $\gamma$ for
	the number of gaps. We give an $O(\gamma n^6)$-time algorithm for this problem.
	\item 
	In Section~\ref{sec: min gaps with throughput requirement} we study the
	variant where we need to minimize the number of gaps under 
	a throughput requirement, namely where either the number of jobs or
	their total weight must meet a specified threshold.
	We show that this problem can be solved in time $O(n^7)$.
	\item 
	In the two problems above, the underlying assumption was that
	it is desirable to have as few gaps as possible. However, in certain applications
	gaps in a schedule may be actually desirable. This motivates the gap
	scheduling model where we wish to \emph{maximize} the number of gaps
	while scheduling all jobs (providing that the instance is feasible).
	We study this problem in Section~\ref{sec: maximizing the number of gaps}, and
	we provide an algorithm that computes an optimal schedule in time $O(n^5)$.
	\item 
	Instead of the total number of gaps, the \emph{size} of gaps may be a
	useful attribute of a schedule. In Section~\ref{sec: minimizing max gap}
	we study the problem where, assuming that the given instance is
	feasible, we want to compute a schedule for which the maximum gap size
	is minimized. We give an $O(n^2\log n)$-time algorithm for this problem.
\end{itemize}
We also consider scheduling problems where jobs have no deadlines. Now all jobs need
to be scheduled. In this model we can of course schedule all jobs in one block, without gaps,
but then some jobs may need to wait a long time for execution. To avoid this,
we will also take into account the flow time measure, where the flow of a job is the time elapsed between
its release and completion times, and we will attempt to minimize either the
maximum flow or the total flow of jobs. We address three problems in this category:
\begin{itemize}
	\item Minimizing total flow time with a budget $\gamma$ for the number of gaps
	(Section~\ref{sec: minimizing total flow time budget for gaps}).
	As we show, this problem can be solved in time $O(n\log n  + \gamma n)$, by exploiting
	the Monge property of the dynamic programming arrays. The running time is in fact
	$O(\gamma n)$ if the jobs are given in sorted order of release times.
	\item Minimizing the number of gaps with a budget for total flow 
	(Section~\ref{sec: minimizing gaps with bound on total flow}). The algorithm
	from Section~\ref{sec: minimizing total flow time budget for gaps} can be adapted to solve
	this problem in time  $O(n\log n + g^\ast n)$, where
	$g^\ast$ is the optimum value. If the jobs are given in sorted order of release times,
	the running time is $O(g^\ast n)$.
	\item Minimizing the number of gaps with a bound on the maximum flow time
	 (Section~\ref{sec: minimizing gaps with bound on max flow}).
	We show that this problem can be solved in
	time $O(n\log n)$, or even $O(n)$ if the jobs are already sorted in order of increasing
	release times.
	\item Minimizing maximum flow time with a budget $\gamma$ for the number of gaps
	(Section~\ref{sec: min max flow bound on gaps}). 
	For this problem we give an	algorithm with running time $O(n\log n)$.
\end{itemize}

Summarizing, for all these problems we provide polynomial-time algorithms, with running times ranging from
$O(n\log n)$ for some problems, to $O(n^7)$ for other. Interestingly, the solutions involve a wide spectrum
of algorithmic techniques, including different dynamic programming formulations and
speed-up techniques based on searching Monge arrays, searching $X+Y$ matrices, and implicit binary search.

As another theme throughout the paper, we draw a connection between gap
scheduling problems that we study and their continuous analogues, which are variants of
hitting set problems for intervals of real numbers. In this continuous model,
each job is represented by an interval between its release time and deadline, and a ``schedule'' 
assigns it to a point in this interval. 
For example, the continuous version of the minimum-gap scheduling problem is equivalent to
computing a hitting set of minimum cardinality. As it turns out, for some problems, the continuous
variants provide insights leading to more efficient algorithms for the corresponding
discrete versions, while in other problems completely new techniques are needed to solve
the discrete version.


\section{Preliminaries}
\label{sec: preliminaries}
%
%



The time is assumed to be discrete, divided into unit time intervals $[t,t+1)$, for $t = 1,2,...$,
that we call \emph{slots}. We will number these consecutive slots $0,1,...$, and we
will refer to $[t,t+1)$ simply as \emph{time slot $t$}, or occasionally even as \emph{time $t$}.
By $\calJ$ we will denote the instance, consisting of a set of unit-length jobs numbered
$1,2,...,n$, each job $j$ with a given integer release time $r_j$. This
$r_j$ denotes the first slot where $j$ can be executed.

A \emph{schedule} $S$ of $\calJ$ is defined by an assignment of jobs to time slots such 
that (i) if a job $j$ is assigned to a slot $t$ then $t\ge r_j$, and
(ii) no two jobs are assigned to the same slot. If $j$ is assigned to slot $t$
in a schedule $S$ then we say that it is \emph{scheduled} or \emph{executed} at $t$.
In most scheduling problems we assume that all jobs can be scheduled.
In problems that involve throughput we will also consider partial schedules, where only
a subset of the jobs is scheduled (for jobs outside this subset the schedule is undefined).

For a given schedule $S$, time slots where jobs are scheduled are called
\emph{busy}, while all other slots are called \emph{idle}. An inclusion-wise maximal time interval of busy slots is called
a \emph{block} of $S$. An interval between two consecutive blocks in $S$ is called a \emph{gap} of $S$.
Of course, the number of blocks in $S$ is always one more than the number of gaps.


\paragraph{Instances with deadlines.}
In some of the scheduling problems we consider the jobs in $\calJ$ will also have specified deadlines.
The \emph{deadline} of job $j$ is denoted $d_j$, is assumed to be integer, and it is the 
last slot where $j$ can be scheduled. (Thus it may happen that $d_j = r_j$, in which case
$j$ can only be executed in one slot.)

For instances with deadlines, we can restrict our attention to schedules $S$ that satisfy the
\emph{earliest-deadline-first property (EDF)}: at any time $t$, either
$S$ is idle at $t$ or it schedules a pending job with the earliest
deadline. Using the standard exchange argument, any schedule can be
converted into one that satisfies the EDF property and has the same set of busy slots.

Without loss of generality, we can make the following assumptions about $\calJ$:
\begin{description}
	\item{(i)} $r_j \le d_j$ for each $j$, 
	\item{(ii)} all jobs are ordered according to deadlines, that is $d_{1}\le \ldots\leq d_{n}$, 
	\item{(iii)} all release times are distinct and all deadlines are  distinct, and
	\item{(iv)} $\calJ$ is feasible (that is, all jobs can be scheduled).   
\end{description}
The validity of assumptions~(i) and~(ii) is trivial. Assumption~(iv) follows
immediately from~(iii), as we can simply schedule each job at its release time.
Therefore we only need to justify~(iii).

To show that assumption~(iii) is valid, we modify the original instance as follows:
If two release times are equal, say when $r_i = r_j$ and $d_j \le d_i$ for $i\neq j$,  then we let $r_i = r_i + 1$. 
Symmetrically, if $d_i = d_j$ and $r_i \le r_j$ then we let $d_i = d_i-1$. 
If this change produces a job $i$ with $d_i < r_i$, then job $i$ cannot of course be
scheduled. For problems where the feasibility is a requirement, we can then
report that the instance is not feasible. For other problems,
we can remove this job $i$ from the instance altogether. 
The correctness of assumption~(iii) can be justified using a standard exchange argument
that we formalize by proving the lemma below.
(Schedules considered in this lemma are allowed to be partial.) 


\begin{lemma}\label{lem: assumption (iii) is valid}
Let $\calJ'$ be the instance obtained by modifying a given instance $\calJ$
as explained above, and let $X$ be some set of time slots. Then
$\calJ$ has a schedule $S$ whose set of busy slots is $X$ 
if and only if $\calJ'$ has a schedule $S'$ whose set of busy slots is $X$. 
\end{lemma}

\begin{proof}
We now justify Lemma~\ref{lem: assumption (iii) is valid}. It is sufficient
to consider only the case when $\calJ'$ is obtained from $\calJ$ by modifying just
one job, as then we can apply the lemma repeatedly.  So suppose that we have two different jobs $i,j$ in $\calJ$ with 
$r_i = r_j$ and $d_j \le d_i$, and that $\calJ'$ is obtained
from $\calJ$ by replacing $i$ by $i'$ such that  $r_{i'} = r_i + 1$ and $d_{i'} = d_i$.

$(\Leftarrow)$ This implication is trivial, because any schedule $S'$ of
$\calJ'$ gives us a schedule $S$ of $\calJ$ with the same set of busy slots by
simply replacing $i'$ by $i$ (if $i'$ is used at all).

$(\Rightarrow)$ Consider a schedule $S$ of some subset of $\calJ$
in which $X$ is the set of busy slots. If $i$ is not scheduled in $S$ then we can simply use $S' = S$.
If $i$ is scheduled in $S$ at slot other than $r_i$, then we can obtain $S'$
by replacing $i$ by $i'$. The last case is when $i$ is scheduled at a slot $r_i$ in $S$.
If $j$ is scheduled in $S$ as well then we obtain $S'$ by swapping $i$ and $j$ in $S$
and then replacing $i$ by $i'$, with $i'$ scheduled where $j$ was scheduled in $S$.
On the other hand, if $j$ is not scheduled in $S$, then we obtain $S'$ by replacing $i$ by $j$ which is scheduled at $r_j=r_i$.
\qed  
\end{proof}

To implement the modification of the instance outlined before Lemma~\ref{lem: assumption (iii) is valid}, 
when we adjust the release times we can process them
in increasing order to facilitate finding equal release times. Each job's release time can be
incremented at most $n$ times, and maintaining the ordering will introduce a logarithmic overhead.
Deadlines can be processed in the symmetric way.
Then the overall running time to modify the instance will be $O(n^2\log n)$. 
Thus this preprocessing does not affect the  overall running time of our algorithms for instances with
deadlines (that all have running time at least this large).


\paragraph{Instances without deadlines.}
For schedules involving objective functions other than throughput, we can
assume that the jobs are ordered according to non-decreasing release times.
For the total-flow objective function we can also assume that all release times
are different. The reason is that, although modifying the release times may change the total flow value 
(see the definition of the flow time in Section~\ref{sec: minimizing total flow time budget for gaps}, paragraph $1$), 
this change will be uniform for all schedules, so the schedule's optimality will not be affected. 
The appropriate modification of release times can be achieved in time $O(n\log n)$ as follows:
First, sort all jobs in order of release times, so that $r_1\le r_2 \le ... \le r_n$.
Process them in this order. Providing that the new release times $r'_1 < r'_2 < ... < r'_{j-1}$
of jobs $1,2,...,j-1$ are already computed, let the new release time of job $j$
be $r'_j = \max(r'_{j-1}+1,r_j)$. If the jobs are already given in the sorted order, this
process will in fact take time $O(n)$. Thus the running times of our algorithms are not affected by
this preprocessing. 

We remark that modifying release times may affect the maximum flow values non-uniformly (that is, differently for different schedules), 
so we will not be using the
assumption about different release times in Sections~\ref{sec: minimizing gaps with bound on max flow}
and~\ref{sec: min max flow bound on gaps}, where maximum flow of jobs is considered.


\paragraph{Shifting blocks.}
To improve the running time,
some of our algorithms use assumptions about possible locations of the blocks in an optimal schedule.
The general idea is that each block can be shifted, without affecting the objective function,
to a location where it will contain either a deadline or a release time. The following 
lemma (that is implicit in~\cite{Baptiste-etal-07}) is useful for this purpose.
We formulate the lemma for leftward shifts; an analogous lemma can be formulated for
rightward shifts and for deadlines instead of release times.


\begin{lemma}\label{lem: shift}
Assume that all jobs in the instance have different release times.
Let $B = [u,v]$ be a block in a schedule such that the job scheduled at $v$ has release time
strictly before $v$.
Then $B$ can be shifted leftward by one slot, in the sense that the jobs in $B$ can be scheduled
in the interval $[u-1,v-1]$.
\end{lemma}

\begin{proof}
We construct a sequence of job indices $i_1,i_2,...,i_q$ such that $i_1$ is the job scheduled at $v$, 
each job $i_b$, for $b = 2,3,...,q$, is scheduled in $B$ at the release time $r_{i_{b-1}}$ of the previous job
in the sequence, and $r_{i_{q}} < u$. This is quite simple: As mentioned earlier, 
we start by letting $i_1$ be the job scheduled at $v$. Suppose that for some $c \ge 1$ 
we have already chosen jobs $i_1,i_2,...,i_c$ such that $i_c$ is scheduled in $B$
and each $i_b$, for $b = 2,3,...,c$, is scheduled at $r_{i_{b-1}}$.
The choice of this sequence implies that $r_{i_c} < r_{i_{c-1}} < ... < r_1 = v$.
If $r_{i_c} < u$, we let $q =c$ and we are done. So suppose that $r_{i_c} \ge u$.
Since all release times are different, we have $r_{i_c} < r_{i_{c-1}}$. We then take $i_{c+1}$ to be the job 
scheduled at $r_{i_c}$. By repeating this process, we obtain the desired sequence.

Given the jobs $i_1,i_2,...,i_q$ from the previous paragraph, we can modify the
schedule by scheduling $i_q$ at time $u-1$, and scheduling each
 $i_b$, $b = 1,2,...,q-1$ at $r_{i_b}$. This will result in shifting $B$ to the left
by one slot, proving the lemma. 
\qed
\end{proof}


\paragraph{Interval hitting.}
For some of our scheduling problems it is useful to consider their ``continuous'' analogues
obtained by assuming that
all release times and deadlines are spread very far apart; thus in the limit we can
think of jobs as having length $0$. Each $r_j$ and $d_j$ (if deadlines are in the
instance) is a point in time, and to ``schedule'' $j$ we assign it to a point in 
the interval $[r_j,d_j]$. Two jobs that would be assigned to consecutive slots
in a discrete schedule will then end up being on the same point.
This continuous problem is then equivalent to  computing a hitting set for a given collection of intervals
on the real line, with some conditions involving gaps in-between its consecutive points.

More formally, in the hitting-set problem we are given a collection
of intervals $I_j = [r_j,d_j]$, where $r_j,d_j$ are real numbers. Our objective is to compute a set
$H$ of points such that $H\cap I_j\neq\emptyset$ for all $j$.
This set $H$ is called a \emph{hitting set} of the intervals $I_1,I_2,...,I_n$.
(This formalism corresponds to scheduling problems with deadlines and where all jobs need
to be scheduled; it can be easily adapted in a natural way to other variants that we study,
when jobs may not have deadlines, or when some jobs do not need to be scheduled.) 

If $H$ is a hitting set of intervals $I_1,I_2,...,I_n$, then for each $j$ we can pick
a \emph{representative} $h_j\in H\cap I_j$. Sorting these
representatives from left to right, $h_{i_1} \le h_{i_2} \le ... \le h_{i_n}$,
the intervals between consecutive representatives are called \emph{gaps} of $H$. The length of the gap
between $h_{i_b}$ and $h_{i_{b+1}}$ is $h_{i_{b+1}} - h_{i_b}$.

For each gap scheduling problem we can then consider the corresponding
hitting-set problem. For example, minimizing the number of gaps in a schedule
translates into a minimum-cardinality hitting set for a collection of intervals. It is well known
(folklore) that this problem can be solved with a greedy algorithm in time
$O(n\log n)$: Initialize $H=\emptyset$. Then, going from left to right, at each step locate the earliest-ending
interval $I_j$ not yet hit by the points in $H$ and add $d_j$ to $H$.

These interval-hitting problems are conceptually easier to deal with than their discrete counterparts.
As we show, some algorithms for interval-hitting problems extend to their
corresponding gap scheduling problems, while for other these discrete variants require
different techniques.


\section{Maximizing Throughput with Budget for Gaps}
\label{sec: maximizing throughput with budget for gaps}
%
%



In this section we consider a variant of gap scheduling where we 
want to maximize throughput (that is, the number of scheduled jobs), given a budget $\gamma$ for the number of gaps. 
We first show that the continuous version of this problem can be solved in time
$O(\gamma n^2)$. For the discrete case we give an algorithm with running time $O(\gamma n^6)$.


\paragraph{Continuous case.}
Formally, the continuous variant of the problem is defined as follows. We are given a collection
of intervals $I_j = [r_j,d_j]$, $j=1,2,...,n$ and a positive integer $\xi \le n$. The objective
is to compute a set $H$ of at most $\xi$ points that hits the maximum
number of intervals, where a point is said to \emph{hit} a set if it belongs to this set. 
(Here, $\xi$ corresponds to the number of blocks, so its value is one more than the number of gaps.)
Without loss of generality we only need to consider sets $H\subseteq \braced{d_1,d_2,...,d_n}$ and we can assume that all
release times and deadlines are different.

There is a simple dynamic-programming algorithm for this problem that works as follows.
Order the intervals according to deadlines, that is $d_1 < d_2 < ... < d_n$.
For $h = 1,2,...,\xi$ and $b = 1,2,..., n$, let
$\Thput_{b,h}$ be the maximum number of input intervals that can be hit by a subset
$H \subseteq \braced{d_1,d_2,...,d_b}$ such that $|H|\le h$ and $d_b\in H$. 
For all $b$, we first initialize $\Thput_{b,1}$ to be the number of intervals that contain $d_b$.
Similarly, for all $h$, we let $\Thput_{1,h}$ to be the number of intervals that contain $d_1$. 
Then, for all $h = 2,3,...,\xi$ and $b = 2,3,...,n$, we can compute $T_{b,h}$ using the recurrence:
\begin{equation*}
\Thput_{b,h} = \max_{a<b} \braced{\Thput_{a,h-1} + \Delta_{a,b}},
\end{equation*}
where $\Delta_{a,b}$ is the number of intervals $I_i$ such that
$d_a < r_i \le d_b \le d_i$, namely the intervals that are hit by $d_b$ but not by $d_a$.
The output value is $\max_b\Thput_{b,\xi}$.

With a bit of care, all values $\Delta_{a,b}$ can be pre-computed in time $O(n^2)$:
First sort all release times and deadlines. For each $a$, consider only
intervals $I_i$ to the right of $d_a$, namely those with $r_i > d_a$. We will make a sweep through release times and deadlines,
starting at $d_a$, and for each visited point counting the number of intervals hit by this point. 
We start with $x=d_a$ and with a counter $q$ initialized to $0$.
Then iteratively increment $x$
to the next release time or deadline, whichever is earliest. At each step update $q$,
by increasing it if the new point is a release time and decreasing it if the current point is a deadline.
If the new point is $x = d_b$, record the value of $q$ as $\Delta_{a,b}$. This sweep costs time $O(n)$.

This gives us an algorithm with running time $O(\xi n^2)$, because we have
$O(\xi n)$ values $\Thput_{b,h}$ to compute, each computation taking time $O(n)$. 

\smallskip
\emph{Note:}
As we found out after completing the initial version of this manuscript, an algorithm with the same complexity was
given earlier in~\cite{Jansen_etal_disjoint_clique_97}. We have decided to
retain the above solution in the paper as it provides useful context for the discrete case
considered next, accentuating the contrast between the continuous and discrete variants.
Also, recently Damaschke~\cite{Damaschke_refined_17} gave a more efficient
algorithm for the special case when the interval graph induced by intervals $I_1,I_2,...,I_n$ is sparse.


\paragraph{Discrete case.}
For the discrete case, when we schedule unit jobs, a more intricate
dynamic programming approach is needed. 
The fundamental idea of our approach is similar to that in~\cite{Baptiste:min-idle-periods,Baptiste-etal-07,Baptiste_etal_polynomial_12}.

A rough intuition here is that scheduling some jobs with short spans, which are more restricted,
may create a lot of gaps. (A span of job $j$ is $d_j-r_j+1$, the length of the interval where it can be scheduled.)
We would like to distribute jobs with longer spans, as many as possible, to fill many of these gaps. The remaining gaps may be then
filled with jobs that have even longer spans, and so on. Figure~\ref{fig: min number of gaps example}
shows an example of an instance and a schedule that maximizes throughput for the budget of $2$ gaps.

\begin{figure}[ht]
\centering
\includegraphics[width=4in]{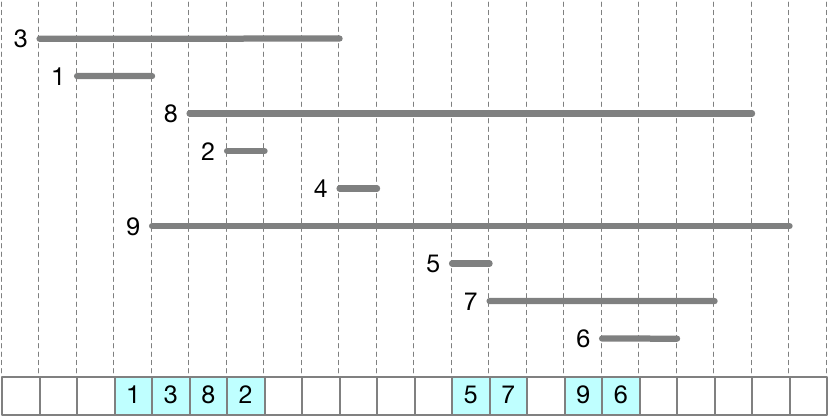} 
\caption{An example of an instance with $n=9$ jobs and its schedule with maximum throughput of $8$ for the budget of $2$ gaps. 
		(There are other optimal schedules.)
		Each job $j$ is represented by a horizontal line segment starting at slot $r_j$ and ending at slot $d_j$.	
}
\label{fig: min number of gaps example}
\end{figure}

Denote by $\calJ$ the set of jobs on input, ordered by deadlines, that is $d_1 < d_2 < ... < d_n$.
(In Section~\ref{sec: preliminaries} we showed that we can assume all deadlines to be different.)
For each job $k$ and times $u\le v$, let $\calJ_{k,u,v}$ denote the sub-instance of $\calJ$ that consists
of all jobs $j\in\braced{1,2,...,k}$ that satisfy $u \le r_j \le v$.
Define $\Thput_{k,u,v,g}$ to be the maximum number of jobs
from $\calJ_{k,u,v}$ that can be scheduled in the interval $[u,v]$ with the number of 
gaps not exceeding $g$. Here, the initial and final gap (between $u$ and the first
job, and between the last job and $v$) are also counted, if present.

To derive a recurrence for $\Thput_{k,u,v,g}$ we reason as follows. If $\calJ_{k,u,v} = \emptyset$ then
$\Thput_{k,u,v,g} = 0$. If $\calJ_{k,u,v} \neq \emptyset$ and $k\notin \calJ_{k,u,v}$
then $\Thput_{k,u,v,g} = \Thput_{k-1,u,v,g}$. So for the rest of the derivation assume that $k\in \calJ_{k,u,v}$.

Consider an optimal schedule $S$ for $\calJ_{k,u,v}$, that is the one that realizes $\Thput_{k,u,v,g}$. 
If $k$ is not scheduled by $S$, then
$\Thput_{k,u,v,g} = \Thput_{k-1,u,v,g}$. In the remaining cases we assume that $k$ is scheduled by $S$,
say at time $t$, where $u \le r_k \le t\le \min(v,d_k)$.    

\begin{figure}[ht]
\centering
\includegraphics[width=5.25in]{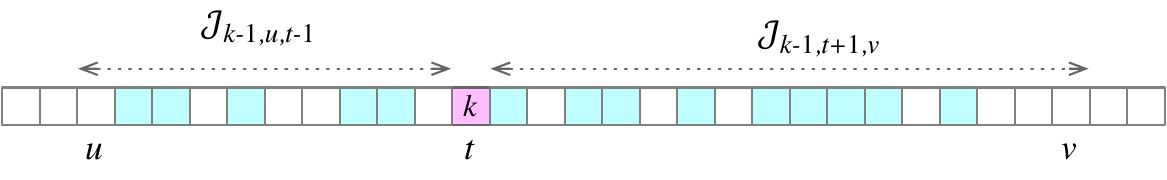} 
\caption{An illustration of the recurrence for $\Thput_{k,u,v,g}$.
}
\label{fig: max throughput gap budget recurrence}
\end{figure}


Naturally, all jobs from  $\calJ_{k-1,t+1,v}$ that are scheduled by $S$ are scheduled in $[t+1,v]$. 
As explained in Section~\ref{sec: preliminaries}, we can assume that $S$ has the EDF property.
Thus no job from $\calJ_{k-1,u,t-1}$ can be scheduled in $[t+1,v]$ because such a job has an earlier
deadline than $k$ and so it cannot be pending in $S$ at time $t$. 
So all jobs from $\calJ_{k-1,u,t-1}$ that are scheduled by $S$ are scheduled in $[u,t-1]$. 
Further, for the same reason, if there is a job in $\calJ_{k,u,v}\setminus\braced{k}$
released at time $t$ then it cannot be scheduled by $S$. (In fact, we can assume that such job
does not exist, because otherwise we could swap it with $k$, as $k$'s deadline is larger. But we do not use this observation
in the algorithm.)

The above paragraph gives us the optimal substructure property needed for a dynamic-programming formulation. 
Specifically, using the optimality of $S$  and letting $h$ be the number of gaps in $[u,t-1]$ in $S$, 
we have that the portion of $S$ in $[u,t-1]$ is a schedule of $\calJ_{k-1,u,t-1}$ with at most
$h$ gaps and maximum throughput, and the portion of $S$ in $[t+1,v]$ is a schedule of
$\calJ_{k-1,t+1,v}$ with at most $g-h$ gaps and maximum throughput. 
(See Figure~\ref{fig: max throughput gap budget recurrence} for illustration.)
Therefore $\Thput_{k,u,v,g} = \Thput_{k-1,u,t-1,h} + \Thput_{k-1,t+1,v,g-h}+1$. 
 
Overall, for $k\in \calJ_{k,u,v}$, the argument above gives us the following formula for $\Thput_{k,u,v,g}$: 
\begin{equation}
\Thput_{k,u,v,g} \;=\;
	\max\braced{\begin{array}{c} 
				\Thput_{k-1,u,v,g}  
				\\ [3pt]
				\displaystyle
				\max_{\begin{subarray}{c}
				r_k\le t \le\min(d_k, v)\\
				0 \le h \le g
				\end{subarray}} 
				\braced{ \Thput_{k-1,u,t-1,h}  + \Thput_{k-1,t+1,v,g-h} } + 1 
			\end{array}
			}
\label{eqn: max throughput idea}
\end{equation}
The solution of the original instance $\calJ$ is $\Thput_{n,\rmin-1,d_n+1,\gamma}-2$,
where $\rmin$ is the minimum release time. (Recall that $d_n$ is the maximum deadline, by
the deadline ordering.) We subtract $2$ to account for the initial and final gap which
will always be present in the overall solution for $\calJ_{k,u,v}$ as we start with the interval $[u,v] = [\rmin-1,d_n+1]$.

To achieve polynomial time we still need to somehow limit the ranges of $u$, $v$ and $t$ in~(\ref{eqn: max throughput idea})  
to some polynomial-size domain. This can be achieved using
Lemma~\ref{lem: shift} which implies that we only need to consider schedules in which every block ends at some release time. 

Define $R = \braced{r_i \suchthat 1\le i\le n}$ to be the set of all release times, 
and for any interval $[x,y]$ of integers define $R+ [x,y] =  \braced{r+z \suchthat r\in R \;\&\; z\in [x,y]}$. 
(For $y=x$ we will simplify this notation and write $R+[x]$ instead of $R+[x,x]$.)
Then, by the above paragraph, we can assume that all busy slots are in the set $R+[-n+1,0]$.
The slot $t$ in the bottom option on the right-hand side of recurrence~(\ref{eqn: max throughput idea}) is always busy, and 
the slots $u$ and $v$ are equal, respectively, $t+1$ and $t-1$, so they are either busy or are adjacent to a busy slot.
Therefore we can restrict the ranges of $u$, $v$ and $t$ to the set $R+[-n,1]\cup\braced{d_n+1}$, which has cardinality $O(n^2)$. 
(We need to also include $d_n+1$, which is the value of $v$ in the solution for the whole instance $\calJ$.)
This gives us a bound of $O(\gamma n^5)$ on the number of
values $\Thput_{k,u,v,g}$ to be computed, each requiring time $O(\gamma n^2)$.
Thus the overall running time is $O(\gamma^2 n^7)$.


\emparagraph{A faster algorithm.}
We now show how to improve this running time by two orders of magnitude.
To this end, we further restrict the range of the left endpoint $u$ to the set $R$. This will involve a
slight modification of the recurrence and the instance (adding an artificial ``dummy'' tight job).
The second improvement is obtained by distinguishing two cases, depending on whether
or not $k$ is the last job in the optimal schedule. If $k$ is not last,
we can reduce the range of $t$ to $R+[-1]$, and if $k$ is last then we can eliminate the maximization over $h$.
The details follow.

As a first step, we claim that we can assume that in the original instance $\calJ$
the first job is a tight job separated from the rest of the instance, that is $r_1 = d_1 \le \min_{j \neq 1}r_j - 2$.
If the first job does not satisfy this property, we can simply add such a job, without affecting the asymptotic
running time. The optimal value for the whole instance $\calJ$ will be computed as $\Thput_{n,\rmin,d_n+1,\gamma}-1$,
with $1$ subtracted to account for the final gap. So in this initial recursive call the second parameter is in $R$.
(If job $1$ was artificially added to $\calJ$, the optimal solution for  $\calJ\setminus\braced{1}$ can be recovered from
the optimal solution of $\calJ$ by subtracting $1$, to account for the gap between job $1$ and the rest of the schedule.)

Then we proceed by induction. Consider a sub-instance $\calJ_{k,u,v}$, with $u\in R$, for which we want to compute $\Thput_{k,u,v,g}$.
We can assume that $k\in \calJ_{k,u,v}$, as otherwise $\Thput_{k,u,v,g} = \Thput_{k-1,u,v,g}$.
We have two cases, depending on whether $k$ is last or not. 

Suppose that $k$ is not last. In this case we can assume that there is a job scheduled right after $k$, at time $t+1$,
for otherwise we could reschedule $k$ as the first job in the next block, without increasing the number of gaps.
(Here we use the fact that $k$ has maximum deadline in $\calJ_{k,u,v}$.) 
By the EDF property, no scheduled jobs in $\calJ_{k,u,v}\setminus \braced{k}$ are pending
at time $t$. Thus the job scheduled at time $t+1$, say $c$, is scheduled at 
its release time $r_c=t+1$. Therefore in this case we have
$\Thput_{k,u,v,g} = \Thput_{k-1,u,t-1,h}  + \Thput_{k-1,t+1,v,g-h} + 1$ for some $h$ (as in 
recurrence~(\ref{eqn: max throughput idea})), where $u, t+1\in R$, and $t\in R+[-1]$.

Next, assume that $k$ is scheduled last. In this case we do not know how to further restrict the
range of $t$, but instead we can avoid maximization over $h$. The optimal substructure 
property holds here as well, that is the portion of $S$ in the interval $[u,t-1]$ must be
an optimal schedule for the corresponding sub-instance. Thus the recurrence has two sub-cases:
If $t=v$ then there is no final gap and $\Thput_{k,u,v,g} = \Thput_{k-1,u,v-1,g} + 1$.     
Otherwise, there is a final gap and $\Thput_{k,u,v,g} = \Thput_{k-1,u,t-1,g-1} + 1$. 
In both cases the second parameter in the recursive call is $u\in R$ and there is no maximization with respect to $h$.


\medskip
\myparagraph{Algorithm~{\MaxThrpt}.}
As explained above, we assume in the algorithm that $r_1 = d_1 \le \min_{j \neq 1}r_j - 2$.
For all $k = 0,1,...,n$ and time slots $u,v$, where $u\in R$, $v\in R+[-n,1]\cup\braced{d_n}$ and $u\le v$,
we process all instances $\calJ_{k,u,v}$ in order of increasing $k$, and for each
$k$ in order of increasing interval length, $v-u$. For each instance $\calJ_{k,u,v}$
and each gap budget $g=0,1,...,\gamma$ we compute the corresponding value $\Thput_{k,u,v,g}$. 
If some value of $\Thput_{k,u,v,g}$ appears on the right-hand side of the recurrence with
$v$ outside its range (that is when $v\notin R+[-n,1]\cup\braced{d_n}$), then we assume that $\Thput_{k,u,v,g} = -\infty$.

First, if $\calJ_{k,u,v} = \emptyset$, we let $\Thput_{k,u,v,g} =0$. Assume now that $\calJ_{k,u,v} \neq \emptyset$.
If $k\notin \calJ_{k,u,v}$ (which means that $r_k \notin [u,v]$)
then $\Thput_{k,u,v,g} =  \Thput_{k-1,u,v,g}$. Otherwise, we compute
$\Thput_{k,u,v,g}$ using the following recurrence:
\begin{equation}
	\Thput_{k,u,v,g} \;=\;
		\max\braced{\begin{array}{l}
					\Thput_{k-1,u,v,g}  
					\\[5pt] 
					\displaystyle
					\max_{\begin{subarray}{c}
						t \in R'\\
						0 \le h \le g
						\end{subarray}} 
					\braced{ \Thput_{k-1,u,t-1,h}  + \Thput_{k-1,t+1,v,g-h} } + 1
					\\[20pt] 
					\displaystyle
					\max_{t\in R''} 
					\braced{ \Thput_{k-1,u,t-1,g-1}} + 1
					\\[10pt] 
					\Thput_{k-1,u,v-1,g} + 1 \quad\textrm{if}\ d_k\ge v {\hfill}
				\end{array}
				}
\label{eqn: max throughput polynomial} 
\end{equation}
where the ranges of $t$ above are
\begin{align*}
      R' \;&=\; (R+[-1]) \cap [r_k, \min (d_k,v) ]
\\
	  R'' \;&=\; (R+[-n+1,0]) \cap [r_k, \min (d_k,v) ]  
\end{align*}
The algorithm outputs $\Thput_{n,\rmin,d_n+1,\gamma}-1$ as the solution to the whole instance $\calJ$.
(This formula is explained before the statement of the algorithm.)

As discussed earlier, with the above restrictions on $u$, $v$ and $t$, we have
$n$ choices for $u$ and $O(n^2)$ choices for $v$. With $n+1$ choices for $k$ 
and $\gamma+1$ choices for $g$, the size of the $\Thput_{k,u,v,g}$ table is $O(\gamma n^4)$. 
In the above recurrence, in the second option we iterate over up to $n$ choices for $t$ and
$\gamma+1$ choices for $h$, and in the third option we iterate over up to $n^2$ choices for $t$.
So the overall running time is $O(\gamma n^6)$. 

Summarizing, we obtain the following theorem:

\begin{theorem}\label{thm: max throughput}
For any instance $\calJ$ and a gap budget $\gamma\le n$, Algorithm~{\MaxThrpt} in time $O(\gamma n^6)$
computes a schedule of $\calJ$ that has maximum throughput among all schedules with at most $\gamma$ gaps.  	
\end{theorem}


\paragraph{Weighted throughput.}
Theorem~\ref{thm: max throughput}  easily extends to the
model where jobs have non-negative weights and the objective is to maximize
the weighted throughput. The only change is that instead of adding $1$ in the
recurrence for $\Thput_{k,u,v,g}$ we would add the weight of $k$. The running time remains the same.


\section{Minimizing the Number of Gaps with Throughput Requirement}
\label{sec: min gaps with throughput requirement}
%
%



Suppose now that we want to minimize the number of gaps under a throughput requirement, 
that is we want to find a schedule that schedules at least a given number $m \in\braced{0,1,...,n}$ 
of jobs while minimizing the number of gaps. Without loss of generality we can assume that there 
exists a schedule with
throughput at least $m$; in fact, as explained in Section~\ref{sec: preliminaries},
we can even assume that the whole instance is feasible.

We can solve this problem, both the continuous and discrete version,
by leveraging the algorithms from the previous section.
We explain the solution for the continuous variant; the solution of the
discrete case can be obtained in an analogous manner.

Recall that $\Thput_{b,h}$ was defined to be the  
maximum number of intervals that can be hit with a subset of $\braced{d_1,d_2,...,d_b}$ that includes $d_b$
and has cardinality at most $h$. We can use all these values to compute
$\Thput_h$, which is the maximum number of intervals that can be hit with a set of cardinality at most $h$
(without any additional restrictions). Computing these for all $h$ will take time $O(n^3)$.
By definition, we have $\Thput_1 \le \Thput_2 \le ... \le \Thput_n$. Then, given our requirement $m$ on the throughput, we compute the
smallest $h$ for which $\Thput_{h}\ge m$. This $h$ is the output of the algorithm. The total running time will be $O(n^3)$.

\smallskip

An essentially identical scheme will produce an algorithm for the discrete case with running time $O(n^7)$,
giving us the following result.

\begin{theorem}\label{thm: main gaps with throughput requirement}
For any instance $\calJ$ and $m\le n$, the above-described algorithm in time $O(n^7)$
computes a schedule of $\calJ$ that has the minimum number of gaps among all schedules
with throughput at least $m$.  
\end{theorem}

\emph{Comment:} The time bounds for the continuous and discrete versions can be
refined by expressing them in terms of the optimum number $g^\ast$ of gaps.  
This can be achieved by stopping the computation of the recurrence formulas for the smallest $h$ for
which the throughput is reached. For the discrete case, the running time will then be $O(g^\ast n^6)$.


\paragraph{Weighted throughput.}
Similar to the previous section, Theorem~\ref{thm: main gaps with throughput requirement}
also holds for the weighted throughput case, as the algorithm does not depend on the
threshold value $m$ being bounded by $n$.


\section{Maximizing the Number of Gaps}
\label{sec: maximizing the number of gaps}
%
%



In the preceding sections we studied  problems where we were interested in
schedules with as few gaps as possible. However, in some applications, gaps in the
schedule may actually be desirable. This can arise, for example, when the input stream consists 
of two types of jobs, some with high priority and other with low priority. High-priority
jobs are allowed to reserve their slots in advance, while low-priority jobs are executed only
if there are slots available. We can then
schedule high-priority jobs first, and maximizing the number of gaps in their schedule 
would help to improve throughput and latency for low-priority jobs.  
One such specific scenario appears in QoS networks when coordination of access to a Wi-Fi channel 
is implemented using so-called point coordination function (PCF) mechanism~\cite{pcf-wikipedia}.
One of the features of PCF is that it inserts gaps (in our terminology)
into the schedule of high-priority traffic in order to allow low-priority traffic to access the channel.

Thus in this section we will examine the variant of gap scheduling where the objective is to
create as many gaps as possible in the schedule. The continuous version of
this problem is trivial: for any interval $I_j = [r_j,d_j]$ with $r_j = d_j$, we must of course
choose $h_j = r_j$. Each interval $I_j = [r_j,d_j]$ with $r_j<d_j$ can be assigned a unique 
point $h_j\in I_j$. Thus in this section we will focus only on the discrete model.

\smallskip

Specifically, we are again given an instance $\calJ$ with $n$ unit jobs with release times
and deadlines, and we assume that the instance is feasible, that is all jobs can be scheduled.
The objective is to find a schedule for $\calJ$ (with all jobs scheduled) that maximizes the number of gaps.
As before, we will assume that all jobs have different deadlines and different
release times, and that they are ordered according to increasing deadlines, $d_1 < d_2 < ... < d_n$. 
We can  also assume that jobs $1$ and $n$ satisfy
$r_1 = d_1 = \min_{j > 1} r_j -2$ and $d_n = r_n = \max_{j < n} d_j+2$, that is,
they are tight jobs executed at the beginning and end of the schedule, separated by gaps from other jobs.
Such jobs can be added to the instance, increasing the number of gaps uniformly by $2$ for 
all schedules; thus the choice of the optimum schedule is not affected, only its value increases by $2$. 
(This is a technical assumption that allows us to fix the range of the dynamic program below.)
Figure~\ref{fig: max number of gaps} shows an example of an instance $\calJ$ with $n = 10$ jobs and its schedule 
with $7$ gaps.

\begin{figure}[ht]
\centering
\includegraphics[width=3in]{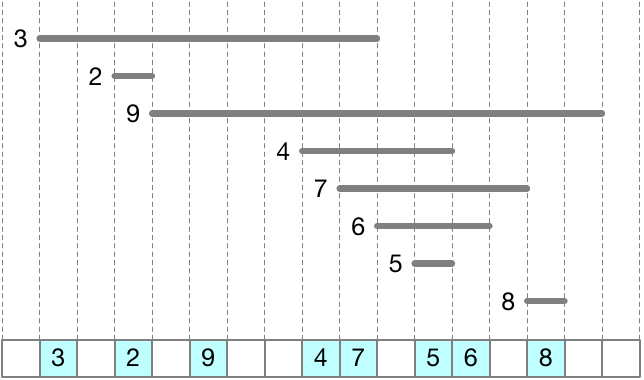} 
\caption{An example of an instance and its schedule with maximum number of gaps, for $n= 10$. 
		Each job $j$ is represented by a horizontal line segment starting at slot $r_j$ and ending at slot $d_j$.
		The special jobs $1$ and $n = 10$ are not shown. In this schedule we have $7$ gaps, which includes
		the gap between jobs $1,3$ and the gap between job $8,10$.		
}
\label{fig: max number of gaps}
\end{figure}


As in Section~\ref{sec: maximizing throughput with budget for gaps}, for any job $k = 1,2...,n$ and two time steps $u \le v$ define
$\calJ_{k,u,v}$ to be the sub-instance of $\calJ$ that consists of all jobs $j\in\braced{1,2,...,k}$ that satisfy $u \le r_j \le v$.
By the assumption about different deadlines, each sub-instance $\calJ_{k,u,v}$ is feasible.
Define $\MGaps_{k,u,v}$ to be the maximum number of gaps in a schedule
of $\calJ_{k,u,v}$ in the interval $[u,v]$. In $\MGaps_{k,u,v}$ we include the extremal gaps in the schedule (if any),
namely the initial gap between $u$ and the first job and the final gap between the last job and $v$.


\begin{lemma}\label{lem: max gaps left adjusted}
For any sub-instance $\calJ_{k,u,v}$ there is a schedule $S$ with $\MGaps_{k,u,v}$
gaps in the interval $[u,v]$ that satisfies the following two conditions:
\begin{description}
	\item{(i)} For any job $j\in \calJ_{k,u,v}$, if $j$ is scheduled at time $S_j$
		then all gaps in the interval $[r_j,S_j]$ have length at most $2$ (including
		the gap between $r_j$ and the first job).
	\item{(ii)} For each block $B$ of $S$, either all jobs in $B$ are scheduled
		at their release times or, assuming that $B$ does not start at $u$,
		the gap immediately to the left of $B$ has length $1$.
\end{description}
\end{lemma}

\begin{proof}
We show that we can modify any schedule $S$ with $\MGaps_{k,u,v}$ gaps
to have properties (i) and (ii).

First, suppose that some job $j$ violates property (i), that is $S$ has
a gap $[x,x']$ such that $r_j \le x < x+2 \le x' \le S_j-1$.  We can then move $j$ to time slot $x+1$. 
Removing $j$ from time slot $S_j$ can decrease the number of gaps at most by $1$
(if $j$ was in a block by itself). Rescheduling $j$ at time $x+1$ will increase
the number of gaps by $1$. Thus overall the number of gaps cannot decrease.

If $S$ has a block $B = [y,y']$ that violates property (ii), choose $j$ to be the first job in $B$
with $S_j > r_j$. Since all release times are 
different, we must have $r_j < y$. We can then move $j$ to slot $y-1$ and,
since the gap that precedes $B$ has length at least $2$, the number of gaps will not decrease.

The two operations above convert the current schedule $S$ into a new schedule $S'$
whose set of busy slots is lexicographically smaller than that of $S$. The number
of gaps in $S'$ is the same or larger than the number of gaps in $S$.
Thus this process must eventually end, producing a schedule that satisfies conditions (i) and (ii).
\qed
\end{proof}

At the very fundamental level, the idea behind our algorithm is similar to that in 
Section~\ref{sec: maximizing throughput with budget for gaps}. We use dynamic programming
to compute all values $\MGaps_{k,u,v}$. Assume that
$k\in \calJ_{k,u,v}$, for otherwise $\MGaps_{k,u,v} = \MGaps_{k-1,u,v}$.
Suppose that, in some optimal schedule $S$ for $\calJ_{k,u,v}$, $k$ is scheduled at some time $t\in[u,v]$. 
Obviously, we have $t\ge r_k \in [u,v]$.
By the EDF property, $t$ itself cannot be a release time of any job in $\calJ_{k,u,v}$ other than $k$.
This property is important for the correctness of our
recurrence, as it implies that $\calJ_{k,u,v}$ can be partitioned into three
disjoint sets: $ \calJ_{k,u,v} = \calJ_{k-1,u,t-1}\cup\braced{k}\cup \calJ_{k-1,t+1,v}$.
Naturally, all jobs in $\calJ_{k-1,t+1,v}$ are scheduled by $S$ in $[t+1,v]$.
Further, using the EDF property again, all jobs in $\calJ_{k-1,u,t-1}$
cannot be scheduled after $t$, so they are all scheduled in $[u,t-1]$.
This implies the following optimal substructure property:
the portion of $S$ in $[u,t-1]$ is an optimal schedule of $\calJ_{k-1,u,t-1}$, 
and the portion of $S$ in $[t+1,v]$ is an optimal schedule of $\calJ_{k-1,t+1,v}$.
We thus conclude that $\MGaps_{k,u,v} = \MGaps_{k-1,u,t-1} + \MGaps_{k-1,t+1,v}$.

Since we do not know $t$ a priori, we can maximize the expression on the right-hand
side over all choices of $t$, giving us the recurrence for $\MGaps_{k,u,v}$ (in the
case when $k\in \calJ_{k,u,v}$):
\begin{equation}
	\MGaps_{k,u,v} \;=\; \max_{\begin{subarray}{c}
									r_k\le t\le \min(v,d_k)\\ t\notin R_{k-1,u,v}
								\end{subarray}}
								\braced{ \MGaps_{k-1,u,t-1} + \MGaps_{k-1,t+1,v} }		
					\label{eqn: max gaps idea}
\end{equation}
where we use notation $R_{k-1,u,v}$ for the set of release times of the jobs in $\calJ_{k-1,u,v}$.
Note that the range of the maximum above is not empty, because $r_k\le \min(v,d_k)$ and $r_k\notin R_{k-1,u,v}$, 
so $r_k$ is a candidate for $t$.
We still need to show that we can reduce the ranges of $u$, $v$ and $t$ in~(\ref{eqn: max gaps idea})
to some polynomial-size domain.

We claim that we only need to consider instances $\calJ_{k,u,v}$ where $u,v\in R+[-1,3n+1]$.
(See Section~\ref{sec: maximizing throughput with budget for gaps} for the definition of sets $R+[x,y]$.)
Indeed, this follows from Lemma~\ref{lem: max gaps left adjusted}(i), which implies that in the
recurrence (\ref{eqn: max gaps idea}) for $\MGaps_{k,u,v}$ we only need to consider slots $t$
between $r_k$ and $r_k + 3n$, inclusive. Thus, in the sub-instances $\calJ_{k-1,u,t-1}$
and $\calJ_{k-1,t+1,v}$ the new arguments $v' = t-1$ and $u' = t+1$ will satisfy
$v', u' \in \braced{r_k-1, r_k, ..., r_k + 3n+1} \subseteq R+[-1,3n+1]$. 
The initial arguments are $r_1$ and $d_n = r_n$, both in $R+[-1,3n+1]$, completing the
proof of our claim. As $| R+[-1,3n+1]| = O(n^2)$, this gives us $O(n^5)$ instances $\calJ_{k,u,v}$ to consider.
For each $\calJ_{k,u,v}$, using Lemma~\ref{lem: max gaps left adjusted}(i),
to compute $\MGaps_{k,u,v}$ it is sufficient to iterate only
over $t = r_k,r_k+1,...,\min(v,d_k,r_k + 3n)$. This would give us the overall running time $O(n^6)$.

Next, we argue that this running time can be further improved to $O(n^5)$. The
general idea is to show that, in essence, the recurrence (\ref{eqn: max gaps idea})
needs to be applied only to $O(n)$ values of $u$. To this end, we modify recurrence~(\ref{eqn: max gaps idea})
as follows:
\begin{equation}
	\MGaps_{k,u,v} \;=\; \max_{\begin{subarray}{c}
									r_k\le t\le \min(v,d_k)\\ t\notin R_{k-1,u,v}
								\end{subarray}} 
				\braced{ \MGaps_{k-1,u,t-1} + \MGaps_{k-1,\mu(t),v} }
					\label{eqn: modified max gaps recurrence modified}
\end{equation}
where $\mu(t)$ is determined based on three cases:  If $\calJ_{k-1,t+1,v} = \emptyset$, let  $\mu(t) = v+1$.
Otherwise, let $\mu' = \min\braced{r_j \suchthat j \in \calJ_{k-1,t+1,v}}$.
If $\mu' = t+1$, let $\mu(t) = t+1$, otherwise let $\mu(t) = \mu'-1$. (Note that
$\mu(t)$ depends also on $v$ and $k$, but we omit these in our notation to reduce clutter.)

We claim that~(\ref{eqn: modified max gaps recurrence modified}) is a correct
recurrence for $\MGaps_{k,u,v}$, providing that $k\in \calJ_{k,u,v}$.
Indeed, from the definition of $\mu(t)$ we have $\calJ_{k-1,t+1,v} = \calJ_{k-1,\mu(t),v}$, and
sub-instance $\calJ_{k-1,\mu(t),v}$ is scheduled inside the interval $[\mu(t),v]$. 
Finally, the optimal schedules of $\calJ_{k-1,t+1,v}$  and
$\calJ_{k-1,\mu(t),v}$ have the same number of gaps. (The reason for distinguishing between the cases
when $\mu'=t+1$ and  $\mu'\neq t+1$ was to take into account the initial gap.)

Using~(\ref{eqn: modified max gaps recurrence modified}), the recurrence
will remain correct if we restrict the range of $u$'s to the set $R+[-1,0]$, whose
cardinality is $O(n)$. Then the total number of instances $\calJ_{k,u,v}$ to consider is
$O(n^4)$, implying the running time of $O(n^5)$. The complete algorithm is described below.


\medskip

\myparagraph{Algorithm~{\MaxGaps}.}
We consider all instances $\calJ_{k,u,v}$, where $u$ and $v$ are time slots
such that $u,v\in [r_1,d_n]$ and $u\le v+1$, and $k$ is either 
a job, that is $k \in \braced{1,2,...,n}$, or $k = 0$.
We process these instances in order of increasing
$k$ and increasing difference $v-u$. For each instance $\calJ_{k,u,v}$,
the value of $\MGaps_{k,u,v}$ is computed as follows.

We first deal with the base case, when $\calJ_{k,u,v} = \emptyset$. Then,
if $u = v+1$ we let $\MGaps_{k,u,v} = 0$, and if $u\le v$ we let $\MGaps_{k,u,v} = 1$. 

So assume now that $\calJ_{k,u,v} \neq \emptyset$, which implies that $u\le v$ and $k\ge 1$. Then,          
if $k\notin  \calJ_{k,u,v}$ we let $\MGaps_{k,u,v} = \MGaps_{k-1,u,v}$.
Otherwise we have $k \in \calJ_{k,u,v}$, in which
case we compute $\MGaps_{k,u,v}$ using the following recurrence:
\begin{equation*}
	\MGaps_{k,u,v} = \max_{\begin{subarray}{c}
										r_k\le t\le \min(v,d_k,r_k+3n)   \\ t\notin R_{k-1,u,v}
						\end{subarray}}
	\braced{ \MGaps_{k-1,u,t-1} + \MGaps_{k-1,\mu(t),v}  }
\end{equation*}
After all values are computed, the algorithm outputs $\MGaps_{n,r_1,d_n}$. 
By the analysis above, we obtain the following theorem.


\begin{theorem}\label{thm: max gaps}
For any instance $\calJ$, Algorithm~{\MaxGaps} in time $O(n^5)$
computes a schedule of $\calJ$ with maximum number of gaps.  	
\end{theorem}


\section{Minimizing Maximum Gap}
\label{sec: minimizing max gap}
%
%



In the earlier sections we focussed on the number of gaps in the schedule. For
certain applications,  the \emph{size} of the gaps is also of interest. 
In this section we will study the problem where the objective is to
minimize the maximum gap in the schedule. Such schedules
tend to spread the jobs more uniformly over the time range and
produce many gaps, which may be useful in applications
discussed in Section~\ref{sec: maximizing the number of gaps}, where a good
schedule should leave some gaps between high-priority jobs, to allow other
jobs to access the processor. This could also be useful in temperature
control of the processor (see the discussion at the end of Section~\ref{sec: final comments}).

The general setting is as before. We have an instance $\calJ$ consisting
of $n$ unit jobs, where job $j$ has release time $r_j$ and deadline $d_j\ge r_j$. 
As explained in Section~\ref{sec: preliminaries}, we can assume that $\calJ$ is feasible.
The objective is to compute a schedule of all jobs that minimizes the maximum gap size.

Interestingly, this problem is structurally different from these in the previous 
sections, because now, intuitively, a good schedule should
spread the jobs more-or-less evenly in time. For example, if we have $n-2$ jobs released 
at $0$, all with deadline $D \gg n$, plus two more tight jobs $1$ and $n$ in time slots $0$ and $D$, respectively,
then we should schedule the non-tight jobs $j = 2,3,...,n-1$
at time slots $\approx (j-1) \frac{D}{n-1}$. In contrast, the algorithms in
Sections~\ref{sec: maximizing throughput with budget for gaps} and~\ref{sec: min gaps with throughput requirement}
attempted to group the jobs into a small number of blocks. 
Similar to the objective in Section~\ref{sec: maximizing the number of gaps}, a schedule
that minimizes the maximum gap size will typically create many gaps, but,
as can be seen in Figure~\ref{fig: example max gaps}, these two objective functions
will in general produce different schedules.


\begin{figure}[ht]
\begin{center}
\includegraphics[width=5in]{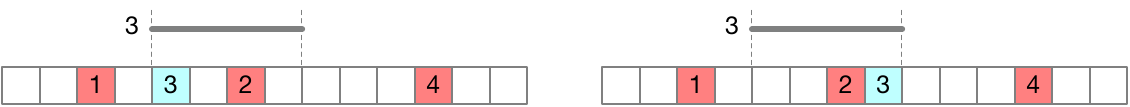}
\caption{An instance with two schedules. Red/dark shaded slots represent tight jobs.
The range of job $3$ is represented by a horizontal segment. The schedule on the left maximizes the
number of gaps. The schedule on the right minimizes the maximum gap. Both
schedules are unique optimal solutions for their respective objective functions.}
\label{fig: example max gaps}
\end{center}
\end{figure}


In this section we give an $O(n^2\log n)$-time algorithm for computing schedules
that minimize the maximum gap. We first give an algorithm for the
continuous model, and then extend it to the discrete model.


\subsection{The Continuous Case}

The continuous analogue of our scheduling problem can be formulated as follows. The input consists of $n$ intervals
$I_1,I_2,...,I_n$. As before, $I_j = [r_j,d_j]$ for each $j$.
The objective is to compute a hitting set $H$ for these intervals that
minimizes the maximum gap between its consecutive points. Another way to think about this problem is as computing a
\emph{representative} $h_j\in H\cap I_j$ for each interval $I_j$. Except for degenerate
situations (two equal intervals of length $0$), we can assume that
all representatives are different, although we will not be using this
property in our algorithm, and we treat $H$ as a multiset.

We order the intervals so that $d_1\le d_2 \le ... \le d_n$. (In this continuous version
we cannot assume all $d_j$'s are different without losing generality.) Further, we only need to
be concerned with sets $H$ that contain $d_1$, because if $H$ contains any points
before $d_1$ then we can replace them all by $d_1$ without increasing the maximum gap in $H$. 
Also, if $\max_i r_i\le d_1$ then there is a singleton hitting set, $H = \braced{d_1}$,
whose maximum gap is equal to $0$. Thus we can also assume
that $\max_i r_i > d_1$, so that we need at least two points in $H$.

Consider first the decision version: \emph{``Given $\lambda > 0$, is there a
hitting set $H$ for $I_1,I_2,...,I_n$ in which all gaps are at most $\lambda$?''}
If $\lambda$ has this property, we will call it \emph{viable}.
We first give a greedy algorithm for this decision version and
then later we show how to use it to obtain an efficient algorithm for the minimization version.


\medskip
\myparagraph{Algorithm~$\AlgViable(\lambda)$.}
We initialize $\circh_1 = d_1$ and $U = \braced{2,3,...,n}$. $U$ represents the set containing  the
indices of intervals that do not have yet representatives selected.
We move from left to right, at each step assigning a representative to one
interval in $U$, placing this representative as far to the right as possible, and we remove this interval from $U$. 

Specifically, at the beginning of a step $j\ge 2$ the current set of representatives is
$\circh_1,\circh_2,...,\circh_{j-1}$, listed in non-decreasing order. In this step we proceed as follows. 
Let $z = \circh_{j-1} + \lambda$. If all $j\in U$ satisfy $r_j > z$, declare failure and return $\FALSE$.
Otherwise, choose $j\in U$ with $r_j \le z$ that minimizes $d_j$ and remove $j$ from $U$.
We now have two cases. If $d_j \le z$, let $\circh_j = d_j$, and otherwise
(that is, when $r_j \le z < d_j$) let $\circh_j = z$. Then increment $j$ and continue.
If the process completes with $U = \emptyset$ (and thus also $j=n$),
return $\TRUE$ and the computed solution $\circH = \smallbraced{\circh_1,\circh_2,...,\circh_n}$.


\medskip

To show correctness of Algorithm~$\AlgViable(\lambda)$,
let $H = \braced{ h_1, h_2, ..., h_n}$ be some solution in increasing order and
with all gaps at most $\lambda$. We show that this solution can be converted into the one computed by our algorithm.
For $j=1$, as we explained earlier, we can assume that $h_1 = d_1$, so  $h_1 = \circh_1$.
Consider the first step when Algorithm~$\AlgViable(\lambda)$ chooses some $\circh_j\neq h_j$.  (If there is no such step, we are done.)
By the choice of $\circh_j$ in the algorithm, we have that $h_j < \circh_j$.  
(Otherwise, either the gap between $h_{j-1}$ and $h_j$ would exceed $\lambda$ or $H$ would not hit $I_j$.)
We can then replace $h_j$ by $\circh_j$ in $H$, without increasing the gap
size to above $\lambda$. This way, we increase the number of steps of
Algorithm~$\AlgViable(\lambda)$ that produce the same representatives as those in $H$.
So repeating this process sufficiently many times eventually converts $H$ into the set $\circH$.

\smallskip

We claim that Algorithm~$\AlgViable(\lambda)$ can be implemented in time $O(n\log n)$. Instead of
$U$, the algorithm maintains a set $U'\subseteq U$ that, when a step $j\ge 2$ starts, consists of indices $i$ for which
$r_i \le \circh_{j-1} + \lambda$ and $I_i$ does not yet have a representative. 
Store $U'$ in a  priority queue with priority values equal to the deadlines.
Then choosing the new interval $I_j$ in the algorithm and removing $j$ from $U'$ takes time $O(\log n)$.
When $j$ is incremented (after adding $\circh_j$ to the solution), the indices of new intervals are inserted
into $U'$ in order of release times (which can be sorted in the pre-processing stage), with each insertion taking time $O(\log n)$.

\medskip

Now, the idea is to use Algorithm~$\AlgViable(\lambda)$ as an oracle in binary search on $\lambda$'s. 
For this to work, we need to be able to efficiently identify a small set of candidate values for the optimal $\lambda$. Let 
\begin{equation*}
	\Lambda \;=\; \braced{ \frac{r_i - d_j}{k} \suchthat k\in \braced{1,2,..., n-1} , i,j\in\braced{1,2,...,n} , r_i > d_j}.
\end{equation*}
Observe that $|\Lambda| = O(n^3)$ and, by our assumption that $\max_i r_i > d_1$, also $\Lambda \neq \emptyset$.

We claim that $\Lambda$ contains the optimal gap length $\starlambda$. The argument is this. 
Consider some hitting set $\starH = \braced{\starh_1,\starh_2,...,\starh_n}$ whose maximum gap is $\starlambda$, 
sorted in non-decreasing order. Choose some maximal (w.r.t. inclusion)
consecutive sub-sequence $\starh_a < \starh_{a+1} < ... < \starh_b$ 
with all gaps equal to $\starlambda$, and suppose that $\starh_a$ is not a deadline. 
Then we can move $\starh_a$ by a little bit to the right without creating a gap
longer than $\starlambda$. Similarly, if $\starh_b$ is not a release time then
we can apply a similar procedure to $\starh_b$ and shift it to the left.
Each such operation reduces the number of gaps of length $\starlambda$. Since $\starlambda$ is optimal,
eventually we must get stuck, meaning that we will find a sub-sequence like the one above with the first
and last indices $a$ and $b$ that satisfy $\starh_a  = d_j$ and $\starh_b  = r_i$, for some $i$ and $j$.
Then we will have $\starlambda = \frac{r_i - d_j}{b-a} \in \Lambda$.

The idea above immediately yields an $O(n^3\log n)$-time algorithm. This algorithm
first computes the set $\Lambda$, sorts it, and then finds the optimal $\lambda$
through binary search in $\Lambda$. Note that the running time is dominated by sorting $\Lambda$. 

\smallskip

We now show that this running time can be improved to $O(n^2\log n)$,
by conducting a more careful search in $\Lambda$ that avoids constructing $\Lambda$ explicitly. 
The basic idea is to use a smaller set $\Delta$ that consists of all values $r_i - d_j$ where $r_i > d_j$.
This set $\Delta$ implicitly represents $\Lambda$, in the sense that it consists of
all numerator values of the fractions in $\Lambda$. More precisely,
each value in $\Lambda$ can be expressed as $x/k$, for some $x\in \Delta$ and $1\le k \le n-1$. 
One can visualize $\Lambda$ by representing such values $x/k$ as points in 2D, 
with the two coordinates representing the values of $x$ and $k$, and point $(x,k)$ representing
$x/k$ (see Figure~\ref{fig: idea behind algorithm minmaxgap}). 
Roughly, the algorithm then finds two consecutive values $v,w$ in $\Delta$ such that
$w/(n-1)$ is viable but $v/(n-1)$ is not. It then finds an index $\kappa$ such that
$v/\kappa$ is viable but $v/(\kappa+1)$ is not. Then the optimum
value of $\lambda$ must be between $v/\kappa$ and $v/(\kappa+1)$. We then
show that there are only $O(n^2)$ such values in $\Lambda$, so by doing a binary
search among these values we can find the optimum $\lambda$ in time $O(n^2\log n)$. 
A detailed algorithm with complete analysis follows.

\begin{figure}[ht]
\centering
\includegraphics[width=3in]{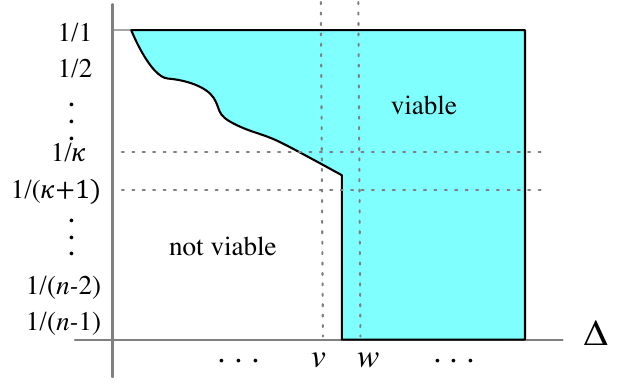} 
\caption{An illustration of the the idea behind Algorithm~$\AlgMinMaxGap$.
Viable fractions in $\Lambda$ are represented by the shaded region.
}
\label{fig: idea behind algorithm minmaxgap}
\end{figure}


\medskip

\myparagraph{Algorithm~$\AlgMinMaxGap$.} 
The algorithm is described below in Pseudocode~\ref{alg:minmaxgap}. In this pseudo-code,
to avoid multi-level nesting, we assume that the algorithm terminates if
the \textbf{return} statement is reached. 


\begin{algorithm}
  \setstretch{1.2}
  \caption{Algorithm~$\AlgMinMaxGap$}
  \label{alg:minmaxgap}
  \begin{algorithmic}[1]
	\State {\myIf} {$\max_i r_i\le d_1$} {\myThen} {\myReturn} $0$
	\State $\Delta \assign \braced{ r_i - d_j \suchthat r_i > d_j, i,j \in \braced{1,2,...,n}}$
	\State sort $\Delta$ in non-decreasing order
	\State {\myIf} {$\AlgViable(\frac{\min(\Delta)}{n-1})$} {\myThen} {\myReturn} $\frac{\min(\Delta)}{n-1}$
	\State $v \assign \max\braced{ x \in \Delta \suchthat \AlgViable(\frac{x}{n-1}) = \FALSE}$
	\State $w\assign \min \braced{ x \in \Delta \suchthat x > v }$
	\State {\myIf} {$\AlgViable(v) = \FALSE$} {\myThen} {\myReturn} $\frac{w}{n-1}$
	\State $\kappa \assign \max\braced{k \in \braced{1,2,...,n-1} \suchthat \AlgViable(\frac{v}{k}) = \TRUE}$
	\State $\Lambda'\assign \braced{ \frac{x}{ \ceiling{ \kappa x/v } }
						\suchthat x\in \Delta \;\textrm{and}\; \frac{v}{\kappa + 1} < \frac{x}{ \ceiling{ \kappa x}/v } \leq \frac{v}{\kappa}
						}\cup\braced{\frac{w}{n-1}}$		
	\State sort $\Lambda'$ in non-decreasing order
	\State \textbf{return} $\min \braced{ \lambda \in \Lambda' \suchthat \AlgViable(\lambda) = \TRUE}$
  \end{algorithmic}
\end{algorithm} 

\bigskip

We now explain the steps in the algorithm and justify correctness and the running time.
First, if $\max_i r_i\le d_1$ then there is a hitting set with all representatives on one point, 
and we return $0$ as the optimum value (Line~1).

Otherwise we have $\max_i r_i > d_1$, that is any hitting set needs at least two points and 
the optimal gap is strictly positive. 
We then compute all positive values $r_i - d_j$, store them in a set $\Delta$ and sort them (Lines~2-3).
This will take time $O(n^2\log n)$.

If $\frac{\min(\Delta)}{n-1}$ is viable (which we check in Line~4), 
then this is the optimal value, since no hitting set can have all
gaps smaller than $\frac{\min(\Delta)}{n-1} = \min(\Lambda)$. 
We can thus now assume that $\frac{\min(\Delta)}{n-1}$ is not viable.

Next, we compute the largest $v\in \Delta$ for which $\frac{v}{n-1}$ is not viable. By the
previous paragraph, such $v$ exists. 
To this end, we can do binary search in the set $\braced{\frac{x}{n-1} \suchthat x \in \Delta}$,
at each step making calls to $\AlgViable()$ to
determine whether the current split value is viable or not.
The binary search will take time $O(n^2\log n)$. 
We also let $w$ to be the next value in $\Delta$ after $v$. (If there is no such value, let $w = \infty$.)

At this point we check whether $v$ is viable. If it is not, it means that for all $x\in\Delta$ with
$x\le v$, all fractions $x/k$, for $k= 1,2,...,n-1$, are not viable as well. 
Then the smallest viable value in $\Lambda$ must be $\frac{w}{n-1}$, so we output $\frac{w}{n-1}$ in Line~7. 
(Note that in this case $w$ must exist, because if $v$ were the largest value in $\Delta$
then $v$ would be viable.)

If $v$ is viable, we compute the largest $\kappa$ for which $v/\kappa$ is viable (Line~8). By the choice of $v$ we have
$\kappa < n-1$. We now also know that the optimal value for $\lambda$ has the form $\frac{x}{k} \in \Lambda$
where $x\in\Delta$, $x\le v$, and
\begin{equation}
	 \frac{v}{\kappa+1} \;<\; \frac{x}{k} \;\le\; \frac{v}{\kappa}. 
	 \label{eqn: min max gap range of x}
\end{equation}
So we only need to search for $\lambda$ among such values.

Next, we define a small set $\Lambda'$ that contains all candidate values from the previous paragraph.
To this end, we claim that for any $x\in \Delta$, if $x\le v$ then there is at most one integer
$k_x\in\braced{1,...,n-1}$ for which condition~(\ref{eqn: min max gap range of x}) holds. 
This follows from simple calculation, as~(\ref{eqn: min max gap range of x}) implies that
\begin{equation*}
			\frac{x}{v}\cdot \kappa \;\le\; k < \frac{x}{v}\cdot \kappa + \frac{x}{v} \;\le\; \frac{x}{v}\cdot \kappa + 1.
\end{equation*}
Thus the only candidate for $k_x$ is $k_x = \ceiling{\frac{x}{v}\cdot \kappa}$.

The above argument gives us that the only candidates for the optimal gap size we need to consider
are all values  $x/k_x$, for $x\in \Delta$ and $x \le v$, plus the value $\frac{w}{n-1}$ that we identified before
as another candidate.
In Lines~9-10 we let $\Lambda'$ be the set of these candidates and we sort them in non-decreasing order.
Finally, we find the smallest viable value in $\Lambda'$. 
As $|\Lambda'| = O(n^2)$, this can be done in time $O(n^2\log n)$ with binary search that calls $\AlgViable()$ for each split value.


\subsection{The Discrete Case}

We now show that Algorithm~$\AlgMinMaxGap$ from the previous section can be adapted
to the discrete case, namely to scheduling unit jobs.

Let $\calJ$ be an instance of unit job scheduling with release times and deadlines.   
As explained in Section~\ref{sec: preliminaries},  we can now assume 
without loss of generality (and in contrast to the continuous case) that
all deadlines are different and sorted in increasing order, $d_1 < d_2 < ... < d_n$.

We treat $\calJ$ as a collection of intervals $I_j = [r_j,d_j]$, $j=1,2,...,n$, and
run Algorithm~$\AlgMinMaxGap$. This will produce a set of
(real-valued) representatives $H = \braced{h_1,h_2,...,h_n}$ for the intervals in $\calJ$. 
(Here $h_j$ denotes the representative of interval $I_j$, so the elements in $H$ may not be in increasing order.)
Let $\lambda$ be the maximum gap between these representatives. Since $\lambda$
is an optimal gap for the continuous variant, $\barlambda = \ceiling{\lambda}-1$
is a lower bound on the optimal gap length for the discrete variant.
(We need to subtract $1$ to account for unit length of jobs.)
It is thus enough to construct a schedule with all gaps of length at most $\barlambda$.

Recall that Algorithm~$\AlgViable(\lambda)$ either assigns jobs to their deadlines
or it spaces consecutive jobs at intervals of $\lambda$ between some deadline and some release time.
As explained before, without loss of generality we can assume that job $1$ is scheduled at $d_1$, and
Algorithm~$\AlgViable(\lambda)$ will in fact produce $h_1 = d_1$. 
If all other $h_i$'s are also deadlines, we are done. Otherwise, the rough idea is
to tentatively assign each job $j$ to $h_j$ (which may not be integral), and then, going from left to right,
gradually shift each job to the first available slot after $h_j$. This does not quite
work, because when several intervals have their representatives in the same slot, this
could force some jobs past their deadlines. So the correct process needs to be more
subtle and allow for some job reordering, as described below.


\medskip
\myparagraph{Procedure~${\AlgAdjust}(\lambda)$.}
We describe how to convert $H$ into a schedule $S$ of $\calJ$. Start by initializing $S_1 = d_1$ and $P=\emptyset$. 
(Set $P$ represents pending jobs that are ``delayed'', namely those
whose representatives' values in $H$ are before or at the current slot.)
Then consider slots $t = d_1+1, d_1 +2 , ...$, one by one. For each such $t$, first add to $P$ all jobs $j$ with $\ceiling{h_j} = t$. 
If $P\neq\emptyset$, choose $j$ to be the job in $P$ with minimum $d_j$, let $S_j = t$, and remove $j$ from $P$.
Then increment $t$ to $t+1$ and continue.
\smallskip


We claim that $S = (S_1,S_2,...,S_n)$ is a feasible schedule. By the way we add
jobs to $P$, if $j\in P$ when we consider slot $t$ then $r_j\le h_j \le t$. Since also
$h_j \le d_j$, each job will be added to $P$ not later than when processing slot $t=d_j$.
Also, the assumption about different deadlines implies (by simple induction) that when we consider a slot $t$ then all jobs in $P$
have deadlines at least $t$; in particular this gives us that no job will miss its deadline. 
Thus $S_j \in [r_j,d_j]$ for all $j\in\calJ$.

Next, we show that the maximum gap size in $S$ is equal to $\barlambda$. Obviously (see above), it
cannot be smaller. To show that it is not larger, consider a tentative assignment $Q = \braced{Q_1,Q_2,...,Q_n}$ of jobs to
slots defined by $Q_j = \ceiling{h_j}$, for all $j\in\calJ$. (This is not a feasible schedule
because it may assign different jobs to the same slot.) We first show that the
maximum gap in this assignment is at most $\barlambda$. Consider two jobs $j$ and $j'$ that are consecutive in $Q$;
that is, $Q_j < Q_{j'}$ and there is no job $\ell$ with $Q_j < Q_\ell < Q_{j'}$.
We can assume that $h_j = \max\braced{ h_\ell \suchthat Q_\ell = Q_j}$ and
$h_{j'} = \min\braced{ h_\ell \suchthat Q_\ell = Q_{j'}}$. Then $j$ and $j'$ are also consecutive in $H$ and
the length of the gap between them is $h_{j'} - h_j \le \lambda$. We then have
\begin{equation*}
Q_{j'} \;=\; \ceiling{h_{j'}} \le \ceiling{ h_j + \lambda} 
		\;\le\; \ceiling{h_j} + \ceiling{\lambda} 
		\;\le\; Q_j + 1 + \barlambda.
\end{equation*}
Thus all gaps in $Q$ are at most $\barlambda$. But all slots of $Q$ are also used by $S$
because, in Procedure~${\AlgAdjust}(\lambda)$, when we consider slot $t\in Q$ set $P$ is not empty.
This implies that the gaps in $S$ are bounded from above by $\barlambda$. We can thus conclude that $S$ is optimal.  

The way we described Procedure~${\AlgAdjust}(\lambda)$, its running time would not be bounded by a function of $n$.
This is easy to fix by skipping all the slots $t$ for which the current set $P$ is empty. Specifically,
we do this: Suppose that when we process a slot $t$ we have $P\neq\emptyset$. If $|P|\ge 2$ then $P$ remains non-empty
after scheduling a job in slot $t$, so in this case we increment $t$ by $1$. Otherwise, we increment it to the
first value $\ceiling{h_j}$ after $t$. This way we will only examine $n$ slots.
With routine data structures, this approach will give us running time $O(n\log n)$.

The discussion above focussed only on computing the optimum gap size. Given this value
and using Algorithm~$\AlgViable()$, one can also compute an actual optimum schedule. Summarizing, we obtain the following theorem.


\begin{theorem}\label{thm: min max gap}
For any instance $\calJ$, Algorithm~$\AlgMinMaxGap$ (adapted for the discrete case, as explained above)
in time $O(n^2\log n)$ computes a schedule of $\calJ$ whose maximum gap value is minimized. 
\end{theorem}



\section{Minimizing Total Flow Time with a Budget for Gaps}
\label{sec: minimizing total flow time budget for gaps}


Unlike in earlier sections, we now consider jobs without deadlines and focus
on the tradeoff between the number of gaps and the delay of jobs. 
Formally, an instance $\calJ$ is given by a collection of $n$ unit length jobs. For each job
$j = 1,2,...,n$ we are given its release time $r_j$. If, in some schedule $S$, job $j$ is executed
at time $S_j$ then $F_j = S_j - r_j$ is called the \emph{flow time of $j$ in $S$}.
We are also given a budget value $\gamma$ for the number of gaps.  The objective is to compute a schedule $S$ for $\calJ$ that minimizes  
the total flow time $F_\total(S) = \sum_j F_j$ among all schedules with at most $\gamma$ gaps.
Figure~\ref{fig: example min total flow} shows an example of an instance and a schedule with two gaps.


\begin{figure}[ht]
\begin{center}
\includegraphics[width=5in]{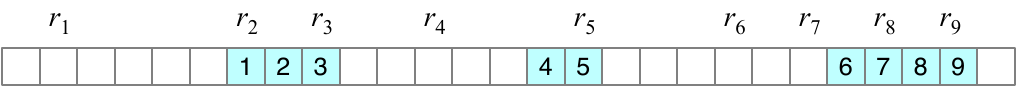}
\caption{An instance and its schedule with two gaps and total flow value 
$F_1 +  ... + F_9 = 5+1+0 + 3 + 0 +3 + 2 + 1 + 0 = 15$.}
\label{fig: example min total flow}
\end{center}
\end{figure}



\paragraph{Continuous case.}
The continuous variant of this problem is equivalent to the \emph{$k$-medians problem} on
a directed line: Given points $r_1,r_2,...,r_n$, find a set $H$ of $k$ points that minimizes the sum 
\begin{equation*}
	\sum_{i=1}^n \;\min_{\begin{subarray}{c}
							h \in H\\ h \ge r_i
							\end{subarray}}
	 				  (h - r_i),        
\end{equation*}
where the $i$th term of the sum represents the distance between $r_i$ and the first point in $H$ after $r_i$.
(Here, the value of $k$ corresponds to $\gamma-1$, the number of blocks in the discrete schedule.)
This is a well-studied problem and it can be solved in time $O(k n)$ if
the points are given in a sorted order~\cite{Woeginger00-monge}. Prior to the work in~\cite{Woeginger00-monge}, 
the undirected case on the line was addressed in~\cite{Hassin_Tamir_location_problems_91,Auletta_etal_resources_98},
and extension to trees have also been studied -- see~\cite{Chrobak_etal_k-median_trees_01}, for example, and references therein. 


\paragraph{Discrete case.}
The discrete case differs from its continuous analogue because
the jobs executed in the same block do not occupy a single point. Nevertheless,
we show that the techniques for computing $k$-medians can be
adapted to minimum-flow scheduling with gaps, resulting in an
algorithm with running time $O(n\log n + \gamma n)$.

Without loss of generality, we assume that all release times are different
and ordered in increasing order, that is $r_1 < r_2 < ... < r_n$. Any instance
can be modified to have this property in time $O(n\log n)$. As explained in
Section~\ref{sec: preliminaries}, this modification changes the flow of all
schedules uniformly, so the optimality is not affected. Sorting the release times is the only
part of the algorithm that requires time $O(n\log n)$; the remaining part will run in time $O(\gamma n)$.

We first give a simple dynamic programming formulation with running time
$O(\gamma n^2)$, and then show how to improve it to $O(\gamma n)$. 
Any schedule with at most $\gamma$ gaps consists of at most $\gamma+1$ blocks.
To reduce the running time, we need to show that these blocks can only
be located at a small number of possible places. For this,
we will need the following lemma, that follows directly from Lemma~\ref{lem: shift} and an exchange argument.


\begin{lemma}\label{lem: max flow opt form}
There is an optimal schedule with the following properties:
(i) all jobs are scheduled in order of their release times,
and 
(ii) the last job of each block is scheduled at its release time.
\end{lemma}

Based on this lemma, each block consists of consecutive jobs, say $i,i+1,...,j$, with the
last job $j$ scheduled at  time $r_j$. The $h$th job of the block is scheduled
at time $r_j - j+h$. So the contribution of this block to the total flow is 
\begin{align*}
W_{i,j} \;=\; \sum_{h=i}^{j} F_h \;&=\; \sum_{h=i}^{j-1} (r_j -j + h - r_h)  
			\\
			  \;&=\; (j-i)r_j - \binom{j-i+1}{2} - R_{j-1} + R_{i-1},
\end{align*}
where $R_b = \sum_{a=1}^b r_a$, for each job $b$.


\paragraph{A simple $O(\gamma n^2)$-time algorithm.}
For each $j=0,1,...,n$, define $\calJ_j$ to be the sub-instance of $\calJ$ consisting of jobs $1,2,...,j$. 
Let $\TotalFlow_{j,g}$ denote the minimum total flow of a schedule for $\calJ_j$ with at most $g$ gaps, where $g\le\gamma$. 
We initialize $\TotalFlow_{0,g} = 0$ for all $g = 0,1,...,\gamma$ and $\TotalFlow_{j,0} = W_{1,j}$ for $j =1,2,...,n$.
Then, for $j = 1,2,...,n$ and $g = 1,2,...,\gamma$, we compute
\begin{equation*}
\TotalFlow_{j,g} \;=\; \min_{1\le i \le j}\braced{ \TotalFlow_{i-1,g-1} + W_{i,j} }.
\end{equation*}
The algorithm returns $\TotalFlow_{n,\gamma}$ as the optimum value for the whole instance $\calJ$.

To justify correctness, we need to explain why the above recurrence holds. Consider a schedule that realizes $\TotalFlow_{j,g}$.
From Lemma~\ref{lem: max flow opt form}, since we are minimizing the total flow, we can
assume that job $j$ is scheduled at $r_j$. Let $i$ be the first job of the last block.
As we calculated earlier, the contribution of this block to the total
flow is $W_{i,j}$. The schedule for the remaining jobs, $1,2,...,i-1$, has at most $g-1$ gaps
and must have optimum total flow time, so (inductively) its total flow time is equal $\TotalFlow_{i-1,g-1}$.

We now consider the running time. All values $W_{i,j}$ can be precomputed in time $O(n^2)$.
We have $\gamma+1$ choices for $g$ and $n+1$ choices for  $j$, so there
are $O(\gamma n)$ values $\TotalFlow_{j,g}$ to compute. Computing each value takes time $O(n)$,
for the total running time $O(\gamma n^2)$.


\paragraph{An $O(\gamma n)$-time algorithm.}
To improve the running time to $O(\gamma n)$, we show that the values $W_{i,j}$
satisfy the Monge property\footnote{For upper triangular matrices this property is often referred to as the quadrangle inequality.
This distinction is only cosmetic, as we can also think of $[W_{i,j}]$ as a full square matrix by filling the lower triangle of the matrix with $\infty$ values.}
(see, for example~\cite{Woeginger00-monge,Burkard_etal_monge_96,Bein_etal_quadrangle_inequality_09}).


\begin{lemma}\label{lem: min total flow quadrangle}
For all $1 \le i \le i' \le j \le j' \le n$, we have
\begin{equation*}
	W_{i,j} + W_{i',j'} \;\le\; W_{i,j'} + W_{i',j}.
\end{equation*}
\end{lemma}

\begin{proof}
It is well known (see~\cite{Burkard_etal_monge_96,Bein_etal_quadrangle_inequality_09}, for example), and easy to prove,
that it is sufficient to prove the inequality in the lemma for $i' = i+1$ and $j' = j+1$, that is	
\begin{equation}
	W_{i,j} + W_{i+1,j+1} \;\le\; W_{i,j+1} + W_{i+1,j}.
		\label{eqn: monge simplified}
\end{equation}
To show~(\ref{eqn: monge simplified}), we compute $W_{i,j} - W_{i+1,j}$ and $W_{i+1,j+1} - W_{i,j+1}$ separately:
\begin{align*}
W_{i,j} - W_{i+1,j} \;&=\; \brackd{ (j-i)r_j - \binom{j-i+1}{2} - R_{j-1} + R_{i-1} }
				\\
					&\quad\quad\quad
						- \brackd{ (j-i-1)r_j - \binom{j-i}{2} - R_{j-1} + R_i }
					\\
					\;&=\; r_j - j + i - r_i,
\end{align*}
and
\begin{align*}
W_{i+1,j+1} - W_{i,j+1} \;&=\; \brackd{ (j-i)r_{j+1} - \binom{j-i+1}{2} - R_{j} + R_{i} }
					\\
					&\quad\quad\quad
						- \brackd{ (j+1-i)r_{j+1} - \binom{j-i+2}{2} - R_{j} + R_{i-1} }
					\\
					\;&=\; - r_{j+1} +j+1-i + r_i.
\end{align*}
Adding these equations, we get
\begin{equation*}
	W_{i,j} + W_{i+1,j+1} - W_{i,j+1} - W_{i+1,j} 
				\;=\; r_j - r_{j+1} + 1 \;\le\; 0,
\end{equation*}
because $r_j < r_{j+1}$, due to our assumption that all release times
are different. This completes the proof of~(\ref{eqn: monge simplified}) and the lemma.
\qed
\end{proof}


\medskip

\myparagraph{Algorithm~$\AlgMinTotFlow$}.
With Lemma~\ref{lem: min total flow quadrangle}, the improved algorithm follows the standard
method of speeding-up dynamic programming by leveraging the Monge property, and is essentially the same as in~\cite{Woeginger00-monge}.
We briefly outline it here for the sake of readers unfamiliar with this approach.
First, we sort the jobs in order of release times. This will cost time $O(n\log n)$.
Unlike in the $O(\gamma n^2)$-time algorithm above, now we will \emph{not} pre-compute all values $W_{i,j}$, as this would cost time $O(n^2)$.
Instead, in time $O(n)$ we precompute only all values $R_b = \sum_{a=1}^b r_a$, for $b = 1,2,...,n$. 
With these values precomputed, we can compute each value $W_{i,j}$ in time $O(1)$ whenever it's needed.
The algorithm then loops on $g = 1,2,...,\gamma$, and for any given $g$ it computes
all $n$ values $\TotalFlow_{j,g}$, for $j = 1,2,...,n$. 
To this end, consider iteration $g$, when the values $\TotalFlow_{j,g-1}$ are already computed for all $j$.
Define an auxiliary function $V_{i,j} = \TotalFlow_{i-1,g-1} + W_{i,j}$. We think of $[V_{i,j}]$ as an implicit matrix
whose values can be each computed in time $O(1)$, when needed. Further, using Lemma~\ref{lem: min total flow quadrangle} it is easy to
show that this matrix $[V_{i,j}]$ also satisfies the Monge property. (The extra $\TotalFlow$-terms in the Monge property for $[V_{i,j}]$ cancel
out, reducing the inequality to Lemma~\ref{lem: min total flow quadrangle}.)
By exploiting this property, in iteration $g$ all minima
$\TotalFlow_{j,g} \;=\; \min_{1\le i \le j} V_{i,j}$, for $j=1,2,...,n$, can be computed in time $O(n)$
using the classical algorithm from~\cite{Aggarwal_etal_matrix_searching_87}.
With the whole matrix $[\TotalFlow_{j,g}]$ computed, the algorithm returns $\TotalFlow_{n,\gamma}$.
The overall running time is $O(n\log n + \gamma n)$.


\begin{theorem}\label{thm: min total flow}
For any instance $\calJ$, Algorithm~$\AlgMinTotFlow$ (as outlined above)
in time $O(n\log n + \gamma n)$ computes a schedule of $\calJ$ that has minimum
total flow among all schedules with at most $\gamma$ gaps.
\end{theorem}



\section{Minimizing Number of Gaps with a Bound on Total Flow}
\label{sec: minimizing gaps with bound on total flow}
%
%



An alternative way to formulate the tradeoff in the previous section would be to find a schedule that
minimizes the number of gaps, given a budget $f$ for the total flow $F_\total$.
This can be reduced to the previous problem by finding the smallest $g$
for which there is a schedule with at most $g$ gaps and total flow at most $f$.
Our solution is the same for both the continuous and discrete versions, so we focus only on the discrete variant. 

Using the notation from the previous section, $\TotalFlow_{n,g}$ represents the minimum total flow of a schedule with at most $g$ gaps.
Then $\TotalFlow_{n,g} = W_{1,n}$, $\TotalFlow_{n,n-1} = 0$, and $\TotalFlow_{n,g}$ is non-increasing as $g$ increases.  
Algorithm~$\AlgMinTotFlow$ computes the values of matrix $[\TotalFlow_{j,g}]$ column by column, that is in order of
increasing $g$. We can then adapt this algorithm to stop as soon as it finds $g$ for which $\TotalFlow_{n,g}\le f$.
Then the minimum number of gaps is $g^\ast = g$. This gives us the following result.


\begin{theorem}\label{thm: min gaps with bound on total flow}
For any instance $\calJ$ and a flow bound $f$, the above modification of Algorithm~$\AlgMinTotFlow$
in time $O(n \log n + g^\ast n)$ computes a schedule of $\calJ$ that minimizes
the number of gaps among all schedules with total flow at most $f$. (Here, $g^\ast\le n$ denotes the number of gaps in the optimum solution.)
\end{theorem}



\section{Minimizing Number of Gaps with a Bound on Maximum Flow}
\label{sec: minimizing gaps with bound on max flow}
%
%



Now, instead of total flow time, we consider the objective function equal 
to the \emph{maximum} flow time, $F_\malymax = \max_j (S_j - r_j)$, that we wish to minimize.
At the same time, we would also like to minimize the number of gaps.
This leads to two optimization problems, by placing a bound on one value and minimizing the other.
In this section we consider  the problem of minimizing the number of gaps when an upper bound on the flow of each
job is given. For this problem,  we give $O(n\log n)$-time algorithm.

Formally, we are given an instance $\calJ$ consisting of $n$ unit jobs with release times and a 
threshold value $f$. The objective is to compute a schedule of $\calJ$ that minimizes
the number of gaps among all schedules with maximum flow time bounded by $f$. 
(If there is no schedule with maximum flow at most $f$, the algorithm should report failure.)
As before, without loss of generality, we can
assume that the jobs are sorted according to their release times, that is
$r_1 \le r_2 \le ... \le r_n$. (As we remarked earlier in Section~\ref{sec: preliminaries}, we cannot now assume
that all jobs have different release times. In fact, the presence of jobs with equal
release times causes the algorithm for the discrete case to be more involved than for the continuous case.)


\paragraph{Continuous case.}
We start by giving an $O(n\log n)$-time algorithm for the continuous case.
Here we are given a collection of $n$ real numbers $r_1,r_2,...,r_n$,
and a number $f$, and we want to compute a set $H$ of minimum cardinality 
such that $\min\braced{h\in H \suchthat h \ge r_i} \le r_i +f$ for all $i=1,2,...,n$.

We show that this can be solved in time $O(n)$, assuming the release times are sorted, $r_1 \le r_2 \le ... \le r_n$.
Indeed, this is very simple, using  a greedy algorithm that computes $H$
in a single pass through the input.  Specifically, initialize $H =  \braced{r_1+f}$.
Then in each step choose $i$ to be smallest index for which $r_i > \max(H)$ and
add $r_i+f$ to $H$. A routine inductive argument shows that the computed set $H$ has
indeed minimum cardinality. The algorithm is essentially a
linear scan through the sorted sequence of release times, so its running time is $O(n)$. With sorting, the time will be $O(n\log n)$.


\paragraph{Discrete case.}
Next, we want to show that we can achieve the same running time for the
discrete variant, where we schedule unit jobs. The greedy single-pass algorithm above does not directly apply because each
point in $H$ corresponds now to a (possibly long) block of jobs, affecting the maximum flow value.

The basic idea of our approach is to think about the problem as the gap minimization
problem with ``virtual'' deadlines, where the virtual deadline of each job $j$ is defined by $ r_j + f$.
We now need to solve the gap minimization problem for jobs with deadlines which, as discussed in the introduction,
can be solved in time $O(n^4)$~\cite{Baptiste:min-idle-periods,Baptiste-etal-07,Baptiste_etal_polynomial_12}.
However, we can do better than this. The instance with deadlines we created satisfies the
``agreeable deadline'' property, which means that the ordering of the deadlines
is the same as the ordering of release times. For such instances a minimum-gap schedule can be computed in time $O(n\log n)$
(see~\cite{Angel_etal_agreeable_energy_12}, for example).
This will thus give us an $O(n\log n)$-time algorithm for gap minimization with a bound on maximum flow.

\smallskip

In the remainder of this section we present an alternative $O(n\log n)$-time
algorithm for this problem, which has the advantage that its running time is
actually $O(n)$ if the jobs are already sorted in non-decreasing order of release times. Besides being of its own interest,
such an algorithm will be useful in the next section.

Let $\calJ$ be the given instance of $n$ unit jobs numbered $1,2,...,n$,
whose release times are ordered in non-decreasing order: $r_1 \le r_2 \le ... \le r_n$. In this ordering the ties are broken arbitrarily.
It is easy to see (by a simple exchange argument) that there is an optimal schedule 
in which all jobs are scheduled in order $1,2,...,n$, and we will only consider such schedules from now on.


\medskip

\myparagraph{Algorithm~{\AlgMinGapMaxFlow}.}
The algorithm has two stages. In the first stage we produce a tentative schedule $Q$ 
by greedily scheduling the jobs from left to right:  Start with $t = r_1$. We have $n$ steps, and
in each step we schedule one job. When step $j$ starts,
jobs $1,2,...,j-1$ will already be scheduled before the current slot $t$. We then schedule $j$ as follows:
If $r_j\le t$, schedule $j$ in slot $t$ and let $t = t+1$; otherwise schedule $j$ in slot $r_j$ and let $t = r_j+1$. 
After we schedule all jobs, we check their flow values. If there is a job in $Q$ with flow
larger than $f$, declare failure (meaning that there is no schedule with maximum flow at most $f$) and stop. 
Otherwise, continue to the next stage.

We now explain the second stage, in which we convert $Q$ into the final schedule $S$. This is accomplished by shifting some jobs to
the right to reduce the number of gaps, without exceeding the maximum flow restriction. The computation consists
of phases numbered $0,1,...,g$, where $g$ is the final number of gaps. In each phase we construct one block of $S$. 
Let $Q^0 = Q$. In general, let $Q^l$ denote the schedule at the start of phase $l$. With $Q^l$ we associate a time slot $v_{l-1}$
which represents the last time slot processed in phases $0,1,...,l-1$. We (artificially) initialize $v_{-1} = r_1-1$. 
The intuition is that, in $Q^l$, the jobs from $Q$ in the time segment $[r_1,v_{l-1}]$ will be rearranged into $l$ blocks, 
while in the time segment $[v_{l-1}+1,\infty]$ the tentative schedule $Q$ will be still unchanged.
Formally, $Q^l$ will satisfy the following invariants (see Figure~\ref{fig: alg mingapmaxflow invariant}):
\begin{description}
\item{(i)} All jobs are scheduled in order of their release times.
\item{(ii)} The jobs scheduled in interval $[r_1,v_{l-1}]$ are exactly the jobs 
			in $\calJ$ released in time segment $[r_1,v_{l-1}]$ and scheduled by $Q$ in this time segment.
\item{(iii)} The jobs in $[r_1,v_{l-1}]$ are scheduled in $l$ blocks
			$B_0, B_1, ..., B_{l-1}$, listed from left to right, where $B_h = [u_h,v_h]$ for $h = 0,1,...,l-1$. 
			In each block $B_h$, at least one job has flow time equal $f$ and all other jobs have flow time at most $f$.
\item{(iv)} In interval $[v_{l-1}+1,\infty]$ schedule $Q^l$ is identical to $Q$.
\item{(v)} Slot $v_{l-1}+1$ is idle in $Q^l$ and is not a release time of any job.
\end{description}
Let $i$ be the first job in $Q^l$ after $v_{l-1}$, scheduled at slot $Q^l_i$. From properties~(iv) and~(v), and from the way
the first stage works, we have $Q^l_i = r_i \ge v_{l-1}+2$.
We start with block $B_l$ initialized as $B_{l} = [u_{l},v_{l}] = [Q^l_i,Q^l_i]$; that is, it consists only of job $i$. 
With $B_{l}$ we associate its maximum flow time value $F(B_{l})$ that is initialized to $Q^l_i - r_i = 0$.
Then, in each step of this phase we will
either shift $B_{l}$ to the right or add another job to it. Specifically, we
do this. If there is a job $j$ scheduled by $Q^{l}$ in time slot $v_{l}+1$, we add this job to $B_{l}$ without changing its schedule,
which means that we increment $v_{l}$, and we update the maximum flow value,
$F(B_{l}) \assign \max\smallbraced{F(B_{l}) , Q^l_j - r_j}$.
Suppose now that there is no job scheduled in slot $v_{l}+1$. If $F(B_{l}) < f$
then we shift $B_{l}$ by $1$ to the right, that is we increment each $u_{l}$, $v_{l}$, and
$F(B_{l})$ by $1$. Otherwise, we end the phase. If the last job in $B_{l}$ is not $n$, we
go to phase $l+1$. If this job is $n$, we are done (and $g = l$), and we return $S = Q^l$.


\begin{figure}[ht]
\begin{center}
\includegraphics[width=6.2in]{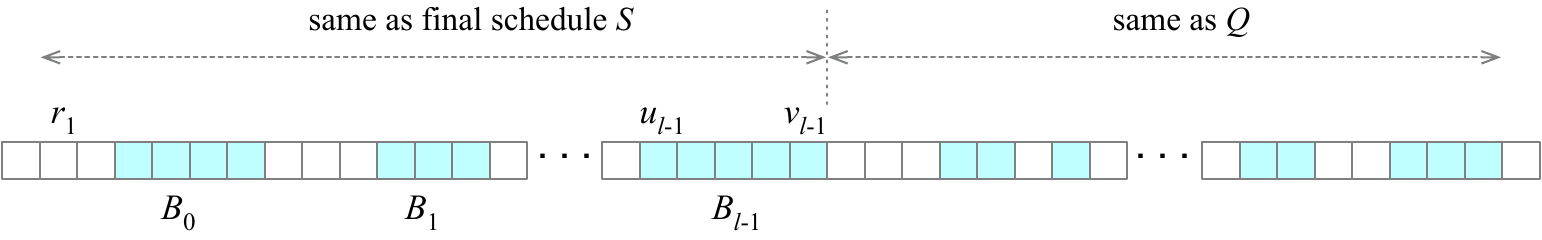}
\caption{Illustration of the invariant for Algorithm~{\AlgMinGapMaxFlow}, showing the structure of schedule $Q^l$ when phase $l$ is about to start.}
\label{fig: alg mingapmaxflow invariant}
\end{center}
\end{figure}


\medskip 
We now argue that Algorithm~{\AlgMinGapMaxFlow} is correct. To this end, we start with the
observation that the tentative schedule $Q$ dominates each other schedule $S'$, in the sense that $Q_j \le S'_j$ for all 
jobs $j$ in $\calJ$. (As explained earlier, we consider only schedules, including $S'$,
where jobs are scheduled in order $1,2,...,n$.) This follows directly from how 
$Q$ is constructed in the first stage, namely that in $Q$ each job $j$
is scheduled at the first idle slot which is not before $r_j$ and is
after the slots of jobs $1,2,...,j-1$. This observation implies that schedule $Q$
minimizes the maximum flow. Therefore if Algorithm~{\AlgMinGapMaxFlow} proceeds to the
second stage, we know that there is a schedule with maximum flow at most $f$.
Further, any such schedule can be obtained from $Q$ by shifting some jobs to the right, preserving the order of jobs.

Consider now $S$. That the flow of all jobs in $S$ is at most $f$ should also
be clear, as  when shifting jobs in the second phase of Algorithm~{\AlgMinGapMaxFlow} 
we explicitly ensure that this condition is preserved. 
Finally, we argue that $S$ minimizes the number of gaps among all schedules with
maximum flow at most $f$. To show this, it is enough to prove that for any two
consecutive blocks $B_{l}$ and $B_{l+1}$ there are two jobs, one in each, that
must be separated by a gap in any schedule with maximum flow at most $f$. To this end, let $p$ be
a job in $B_{l}$ whose flow in $S$ is exactly $f$, that is $S_p = r_p+f$. (Such $p$ exists, by property~(iii).)
Let $i$ be the first job in $B_{l+1}$. Thus job $i-1$ is the last job in $B_{l}$ and
it is scheduled at slot $S_{i-1} = v_{l} = S_{p}+  i-1 -p = r_p + f + i-1-p$. 
By property~(v), we have  $r_i \ge v_{l}+2$. Thus
\begin{equation*} 
	r_i - (r_p + f) \;\ge\;   v_{l}+2  -r_p - f  \;=\;  i - p  +1.
\end{equation*} 
The latest slot when we can schedule job $p$ is $r_p+f$ and the earliest slot when we can schedule
job $i$ is $r_i$. The time segment $[r_p+f,r_i]$ has $r_i - (r_p+f)+1$ slots, which is strictly
greater (by the above inequality) than the number of jobs $i-p+1$ between $p$ and $i$ (inclusive)
that we need to schedule in this segment, so there has to be an idle slot between $p$ and $i$, as claimed.

\smallskip
We now claim that Algorithm~{\AlgMinGapMaxFlow} can be implemented in time $O(n)$. 
The first stage clearly runs in time $O(n)$, so we focus on the second stage.
In our implementation, for each block $B_l$
we have a list of jobs scheduled in this block, in order of release times.
(However, for the jobs in $B_l$ we do \emph{not} keep track of which slot they are scheduled in
during the second stage.
Updating these values after each shift would be too time consuming.) 
Instead of repeatedly shifting $B_l$, we compute the smallest shift value $\delta$ such that,
after shifting $B_l$ by $\delta$, either 
$F(B_l)$ will become equal $f$ or there will be a job scheduled by $Q^l$ right after $B_l$. 
Specifically, if $i$ is the first job scheduled after $B_l$,
at time $Q^l_i$, then the shift value is $\delta = \min(Q^l_i - v_l-1, f-F(B_l))$.
Thus all three values of $u_l$, $v_l$ and $F(B_l)$ are increased by $\delta$. 
After this computation is complete, the
slot $S_j$ of each job $j$ can be computed by adding its index within its block $B_l$ to this
block's start time $u_l$. With these modifications, the running time of the second stage will be $O(n)$. 

\smallskip
Summarizing this section, we obtain the following theorem:         


\begin{theorem}\label{thm: min gaps with bound on max flow}
For any instance $\calJ$ and a flow bound $f$, Algorithm~{\AlgMinGapMaxFlow}
in time $O(n\log n)$ computes a schedule of $\calJ$ that minimizes
the number of gaps among all schedules with maximum flow at most $f$.
If the release times are already sorted, the running time of Algorithm~{\AlgMinGapMaxFlow}
is $O(n)$.
\end{theorem}



\section{Minimizing Maximum Flow with a Budget for Gaps}
\label{sec: min max flow bound on gaps}
%
%



We now consider an alternative variant of the tradeoff between minimizing the maximum flow and the
number of gaps. This time, for a given collection of $n$ unit jobs with release times 
$r_1,r_2,...,r_n$ and a budget $\gamma$, we want to compute a schedule
that minimizes the maximum flow time value $F_\malymax$ and has at most $\gamma$ gaps.
(Recall that $F_\malymax = \max_j F_j$, where $F_j$ is the flow time of job $j$, that is $F_j = S_j - r_j$.)
We can again assume that $r_1 \le r_2 \le ... \le r_n$ and restrict our attention to schedules where jobs are scheduled in order $1,2,...,n$.


\paragraph{Continuous case.}
In the continuous case, $r_1,r_2,...,r_n$ are points on the real line, and we want to compute a set $H$ of at most $k$ points that 
minimizes $F_\malymax(H) = \max_i \min_{x\in H,x\ge r_i} (x-r_i)$.  This is a special case of the $k$-center problem, when the underlying space
is the directed line, which can be solved in time $O(n\log^\ast n)$ if the points are already sorted~\cite{Chrobak_etal_recovery_points_91}.
(The undirected version of this problem has been extensively studied since early 1980's, even for the more general cases of trees~\cite{Chen_etal_k-center_15,Megiddo_Tamir_p-center_83,Frederickson_supply_centers_91,Frederickson_tree_partitioning_91,Frederickson_Zhou_parametric_search_17},
culminating in an $O(n)$-time algorithm~\cite{Frederickson_supply_centers_91}.)
As we do not assume the inputs to be sorted, a simpler $O(n\log n)$-time algorithm that we outline below will be sufficient for our purpose. 
The ingredients for this algorithm are present in various forms in the above cited work on the $k$-center problem,
but we include it here for the sake of completeness, and as a stepping stone to our algorithm for the discrete case.

Similar to our algorithm in Section~\ref{sec: minimizing gaps with bound on total flow},
the high-level idea is based on \emph{parametric search} (see~\cite{Frederickson_tree_partitioning_91,Frederickson_supply_centers_91,Frederickson_Zhou_parametric_search_17},
for example).  It involves binary search for the optimal value
$\starf$ of $F_\malymax(H)$, where at each step of the binary search we use the algorithm from the previous section as an oracle.

For binary search, however, we need a small set of candidate values for $\starf$.
If $H$ is an optimal solution, then, without loss of generality, we can
assume that $H$ contains only release times, since any other point in $H$
can be shifted left until it reaches a release time. Thus we only need
to consider the set $\Phi$ of all values of the form $r_j - r_i$ for $j > i$.
(We will tacitly assume that $F_\malymax(H) > 0$, because verifying whether $F_\malymax(H) = 0$ is easy: just
check if the number of different $r_i$'s is at most $k$.)
Since $|\Phi| = O(n^2)$ and we need to sort $\Phi$ before doing binary search, we would obtain an $O(n^2\log n)$-time algorithm. 

Fortunately, we do not need to construct $\Phi$ explicitly.  Observe that the elements of $\Phi$ can be thought of as
forming an implicit $X+Y$ matrix with sorted rows and columns, where $X$ is the vector of release times and $Y = -X$. 
We can thus use the $O(n)$-time selection algorithm for $X+Y$ 
matrices~\cite{Frederickson_Johnson_x_plus_y_matrices_82,MirzaianA85-selection} to speed up computation. Specifically, at each step
we will have two indices $p, q$, with $1\le p \le q \le n(n-1)/2$, such that the optimal value of $\starf$ is between the
$p$th and $q$th smallest values in $\Phi$, inclusive. If $p=q$, we are done, so
assume that $p < q$. We let $l = \floor{(p+q)/2}$ and
we use the algorithm from~\cite{Frederickson_Johnson_x_plus_y_matrices_82,MirzaianA85-selection} to find the $l$th smallest element in $\Phi$, say $f$. 
We now determine whether $\starf \le f$ by applying the $O(n)$ algorithm from the previous section to answer the query
``is there a set $H$ with $|H|\le k$ and $F_\malymax(H) \le f$?''. If the answer is ``yes'', we let $q=l$, otherwise we let $p = l+1$.
This will give us an algorithm with running time $O(n\log n)$.


\paragraph{Discrete case.}
We now show that we can solve the scheduling variant in time $O(n\log n)$ as well.
The solution is similar to the one for the continuous case, with two modifications.
The first modification concerns the set $\Phi$ of candidate values for the maximum flow. 
We show that $\Phi$ can be still expressed as an $X+Y$ set, for some sets $X$ and $Y$
of small cardinality. The second modification involves using Algorithm~{\AlgMinGapMaxFlow} to answer
decision queries in the binary search, instead of the algorithm for the continuous model.

Without loss of generality, we can restrict our attention to schedules $S$ that have the following structure:
\begin{description}
\item{(i)} Jobs in $S$ appear in order $1,2,...,n$ from left to right. (This assumption was
 already justified earlier).
\item{(ii)} Any block in $S$ contains a job scheduled at its release time. 
 	(Otherwise we can shift this block to the left.)
\item{(iii)} If a job $i$ is scheduled by $S$ in some block $B$, then $r_i$ is either in $B$ or in the gap preceding $B$.  
	(Otherwise, by the ordering of release times and (i), we can assume that $i$ is the first job in $B$.
	We could then move $i$ to the end of the previous block, and repeat this process.)
\item{(iv)}  Any two jobs released at the same time are scheduled in the same block.
	(This follows from (iii).)
\end{description}
To apply search in $X+Y$ matrices, we would like to restrict $X$ and $Y$
to have size $O(n)$. Assumption~(ii) gives us immediately that there is an optimal schedule
where each job is scheduled in a slot in $R+[-n,n]$ (see Section~\ref{sec: maximizing throughput with budget for gaps}), 
but this set has quadratic size, so it's too large for our purpose. 

To construct smaller sets $X$, $Y$, we reason as follows. Consider some optimal schedule $S$.
Choose $i$ to be a job with maximum flow time in $S$, and suppose that $i$ is scheduled
by $S$ in some block $B$. By~(ii), $B$ has a job $j$ scheduled at time $r_j$.
Then the flow time of $i$ can be written as
\begin{align}
 		F_i \;=\; S_i - r_i \;&=\; (r_j + i-j) - r_i 
					\nonumber
						\\
							&=\; (r_j - j) - (r_i - i).  
					\label{eqn: 10 flow time form}
\end{align}
This equation holds no matter whether $j\le i$ or $j > i$.
Now, take $X$ to be the set of all values $r_j-j$ for $j=1,2,....,n$ and
$Y = -X$. We can sort $X$ and $Y$ in time $O(n\log n)$. By~(\ref{eqn: 10 flow time form}),  
we only need to search for the optimal flow value in $\Phi = X+Y$.
Analogously to the continuous case, we perform binary search in $\Phi$, using
the $O(n)$-time algorithm from~\cite{Frederickson_Johnson_x_plus_y_matrices_82,MirzaianA85-selection}
for selection in $X+Y$ matrices and
Algorithm~{\AlgMinGapMaxFlow} as the decision oracle at each step, and since the release times
can be pre-sorted, each invocation of this oracle will take time $O(n)$.
Thus the running time will be $O(n\log n)$. 

The complete algorithm in pseudo-code is given below. In the algorithm we assume that $n\ge 3$ and
$0 \le \gamma \le n-2$, as for $\gamma \ge n-2$ we have $F_\malymax = 0$.
In this pseudo-code, $\AlgMatrixSelect(X,Y, l)$ is a call to an $O(n)$-time
algorithm in~\cite{Frederickson_Johnson_x_plus_y_matrices_82,MirzaianA85-selection} that
finds the $l$th smallest value in the (implicit) matrix $X+Y$.


\begin{algorithm}[H]
  \setstretch{1.2}
  \caption{Algorithm~$\AlgMinMaxflowGap$}
  \label{alg:minmaxflowgap}
  \begin{algorithmic}[1]
	\State {Sort release times so that $r_1 \le r_2 \le ... \le r_n$}
	\State {$X\assign \braced{r_j - j \suchthat j \in\braced{1,2,...,n}}$}
	\State {$Y\assign - X$}
	\State {$p\assign 1$ and $q\assign n$}
	\While{$p< q$}
		\State {$l \assign \floor{(p+q)/2}$}
		\State {$f \assign \AlgMatrixSelect(X, Y, l)$}
		\If{$\AlgMinGapMaxFlow(f) \le \gamma$}
			\State {$q \assign l$}
	    \Else 
			\State {$p\assign l$}
	    \EndIf
    \EndWhile
	\State{\textbf{return} $\AlgMatrixSelect(X, Y, p)$}
  \end{algorithmic}
\end{algorithm}

\smallskip
Summarizing this section, we obtain the following theorem:         


\begin{theorem}\label{thm: min max flow with gap budget}
For any instance $\calJ$ and a gap budget $\gamma$, Algorithm~$\AlgMinMaxflowGap$
in time $O(n\log n)$ computes a schedule of $\calJ$ that minimizes the maximum
flow value among all schedules with at most $\gamma$ gaps.
\end{theorem}



\section{Final Comments}
\label{sec: final comments}
%
%



We studied several scheduling problems for unit-length jobs where the gap structure of the 
computed schedule is taken into consideration. For all problems we considered we 
provided polynomial-time algorithms, with running times ranging from $O(n\log n)$ to $O(n^7)$. 

Many open problems remain. The most intriguing question
is whether the running time for minimizing the number of gaps for unit jobs can be improved to below $O(n^4)$.  
As discussed in Section~\ref{sec: introduction}, this problem is closely related to
energy minimization in the power-down model, and faster algorithms for this problem
would likely also apply to computing minimum-energy schedules.
Speeding up the algorithms in Sections~\ref{sec: maximizing throughput with budget for gaps},
\ref{sec: min gaps with throughput requirement}, \ref{sec: maximizing the number of gaps},
and~\ref{sec: minimizing max gap} would also be of considerable interest.

There is a number of other variants of gap scheduling, even for unit jobs, that we have not
addressed in our paper. Here are some examples:
\begin{itemize}
	\item 
	The problem of maximizing the minimum gap. This is somewhat similar to the
	problem we studied in Section~\ref{sec: minimizing max gap}, but we are not sure
	whether our method can be extended to this model.
	(We remark here that, according to our definition, the minimum gap size cannot be $0$.
	For the purpose of maximizing the minimum gap, one can also consider an alternative
	model where ``gaps'' of size $0$ are taken into account.)
	\item
	The tradeoff between throughput and gap size. Here, one can consider either the lower or upper bound on the gap size.
	\item 
	The tradeoff between flow time (total or maximum) and gap size. For example,
	one may wish to minimize the total flow time with all gaps not exceeding a	specified threshold.
	\item
	The problems of maximizing the number of gaps or minimizing the maximum gap,
	studied in Sections~\ref{sec: maximizing the number of gaps} and~\ref{sec: minimizing max gap},
	were motivated by applications where the schedule for high-priority jobs needs to contain
	gaps where low-priority jobs can be inserted. A more accurate model for such
	applications would be to require that each block is of length at most $b$, for
	some given parameter $b$. Testing feasibility, with this requirement, can be
	achieved in high-degree polynomial time by extending the techniques 
	from~\cite{Baptiste:min-idle-periods,Baptiste-etal-07,Baptiste_etal_polynomial_12} and
	Sections~\ref{sec: maximizing throughput with budget for gaps} and~\ref{sec: minimizing max gap}, 
	but it would be interesting to see whether more efficient solutions exist.
\end{itemize}

A natural extension of our work would be to study variants of gap scheduling for
jobs of arbitrary length, for models with preemptive or non-preemptive jobs.
The algorithm for minimizing the number of gaps, for example, can be
extended to jobs of arbitrary length~\cite{Baptiste-etal-07,Baptiste_etal_polynomial_12} if preemptions 
are allowed, although its running time increases from $O(n^4)$ to $O(n^5)$. 

Another related direction of research would be to focus on the sizes of blocks in the
schedule, or even consider them together with gap sizes. For example, 
schedules with low density (maximum ratio of the number of jobs in an interval to its length)
would be helpful in controlling the processor's temperature during the 
execution~\cite{Chrobak_etal_temperature-aware_2011}, as they include idle time slots 
that allow the processor to cool down between executing consecutive blocks.


\begin{acknowledgements}
	We would like to thank Nael Abu-Ghazaleh for pointing out the connection between
	gap scheduling and wireless channel access scheduling for high- and low-priority
	traffic streams~\cite{pcf-wikipedia}.
\end{acknowledgements}

%
%



\bibliographystyle{spmpsci}
\bibliography{energy,other}

\end{document}